\documentclass{LMCS}

\def\doi{9(2:12)2013}
\lmcsheading%
{\doi}
{1--35}
{}
{}
{Dec.~21, 2012}
{Jun.~28, 2013}
{}

\usepackage[utf8]{inputenc}
\usepackage{hyperref,enumerate}

\usepackage[backgroundcolor=white,bordercolor=red,textsize=small]{todonotes}

\usepackage{amsmath,amssymb,amsthm,stmaryrd}
\usepackage{enumerate}

\usepackage{tikz}
\usetikzlibrary{shapes}
\usetikzlibrary{backgrounds}
\usetikzlibrary{positioning}
\usetikzlibrary{automata}
\usetikzlibrary{arrows}

\newcommand{\lang}{L}
\newcommand{\rank}{\ensuremath{\kappa}}
\newcommand{\cstate}{\ensuremath{c}}
\newcommand{\conc}{\ensuremath{\mathsf{cs}}}

\newcommand{\msg}{\ensuremath{\mathsf{msg}}}
\newcommand{\msgi}{\ensuremath{\mathsf{msg}^{-1}}}
\newcommand{\proc}{\ensuremath{\mathsf{proc}}}
\newcommand{\proci}{\ensuremath{\mathsf{proc}^{-1}}}
\newcommand{\id}{\ensuremath{\mathsf{id}}}
\newcommand{\nx}{\ensuremath{\mathsf{next}}}
\newcommand{\pr}{\ensuremath{\mathsf{prev}}}
\newcommand{\true}{\ensuremath{\mathsf{tt}}}
\newcommand{\N}{\mathbb{N}}
\newcommand{\cA}{\ensuremath{\mathcal{A}}}
\newcommand{\cB}{\ensuremath{\mathcal{B}}}
\newcommand{\Cc}{\ensuremath{\mathcal{C}}}
\newcommand{\cG}{\ensuremath{\mathcal{G}}}
\newcommand{\cM}{\ensuremath{\mathcal{M}}}
\newcommand{\cP}{\ensuremath{\mathcal{P}}}

\newcommand{\pth}[1]{\ensuremath{\left<#1\right>}}
\newcommand{\reach}{\ensuremath{\mathsf{reach}}}
\newcommand{\model}{\ensuremath{\mathsf{mod}}}
\newcommand{\minmodel}[1]{\ensuremath{\llbracket #1\rrbracket}}
\newcommand{\mainstates}{\ensuremath{\mathsf{ms}}}
\newcommand{\poly}[1]{\ensuremath{\mathsf{poly}(#1)}}
\newcommand{\branch}{\ensuremath{b}}

\theoremstyle{plain}
\newtheorem{myclaim}[thm]{Claim}

\begin{document}

\title[PDL with Converse and Repeat for Message-Passing Systems]{Propositional
  Dynamic Logic with Converse and Repeat for Message-Passing Systems}

\author{Roy Mennicke}
\address{Ilmenau University of Technology, Germany}
\email{roy.mennicke@tu-ilmenau.de}

\ACMCCS{[{\bf Theory of computation}]:  Logic---Verification by model checking}

\begin{abstract}
  The model checking problem for propositional dynamic logic (PDL)
  over message sequence charts (MSCs) and communicating finite state
  machines (CFMs) asks, given a channel bound $B$, a PDL
  formula~$\varphi$ and a CFM $\mathcal{C}$, whether every
  existentially $B$-bounded MSC $M$ accepted by $\mathcal{C}$
  satisfies $\varphi$. Recently, it was shown that this problem is
  PSPACE-complete.
  In the present work, we consider CRPDL over MSCs which is PDL
  equipped with the operators converse and repeat. The former enables
  one to walk back and forth within an MSC using a single path
  expression whereas the latter allows to express that a path
  expression can be repeated infinitely often. To solve the model
  checking problem for this logic, we define message sequence chart
  automata (MSCAs) which are multi-way alternating parity automata
  walking on MSCs. By exploiting a new concept called concatenation
  states, we are able to inductively construct, for every CRPDL
  formula $\varphi$, an MSCA precisely accepting the set of models of
  $\varphi$. As a result, we obtain that the model checking problem
  for CRPDL and CFMs is still in PSPACE.
\end{abstract}

\maketitle

\section{Introduction}

Automatic verification is the process of translating a computer system to a
mathematical model, formulating a requirements specification in a formal
language, and automatically checking the obtained model against this
specification. In the past, finite automata, Kripke structures, and Büchi
automata turned out to be suitable formalisms to model the behavior of complex
non-parallel systems. Two of the most common specification languages are the
temporal logics LTL \cite{DBLP:conf/focs/Pnueli77} and CTL \cite{DBLP:conf/lop/ClarkeE81}.
After deciding on a modeling and a specification formalism, automatic
verification melts down to the \emph{model checking problem}: Given a
model~$\cA$ with behavior $\lang(\cA)$ and a specification $\varphi$
representing the expected behavior $\lang(\varphi)$, does
$\lang(\cA)\subseteq \lang(\varphi)$ hold?

Distributed systems exchanging messages can be modeled by
communicating finite-state machines (CFMs) which were introduced in
\cite{DBLP:journals/jacm/BrandZ83}. A CFM consists of a finite number
of finite automata communicating using FIFO channels.  Each run of
such a machine can be understood as a message sequence chart
(MSC). The latter is an established ITU standard and comes with a
formal definition as well as a convenient graphical notation.  In a
simplified model, an MSC can be considered as a structure consisting
of send and receive events which are assigned to unique processes
where the events of each process are linearly ordered. For every send
event, there exists a matching receive event and vice
versa. Unfortunately, the model checking problem for CFMs is
undecidable even for very simple temporal logics -- this is a direct
consequence of the undecidability of the emptiness problem for
CFMs. One solution to this problem is to establish a bound $B$ on the
number of messages pending on a channel. The bounded model checking
problem of CFMs then reads as follows: given a channel bound $B$, a
specification~$\varphi$ and a CFM $\Cc$, does every existentially
$B$-bounded MSC~$M$ accepted by $\mathcal{C}$ satisfy $\varphi$? An
existentially $B$-bounded MSC is an MSC which admits an execution with
$B$-bounded channels.  Using this approach several results for
different temporal logics were obtained in
\cite{DBLP:conf/fsttcs/MadhusudanM01,DBLP:journals/jcss/GenestMSZ06,DBLP:journals/iandc/GenestKM06,DBLP:journals/corr/abs-1007-4764}.

In \cite{DBLP:journals/corr/abs-1007-4764}, a bidirectional
propositional dynamic logic (PDL) was proposed for the automatic
verification of distributed systems modeled by CFMs. This logic was
originally introduced by Fischer and Ladner
\cite{DBLP:journals/jcss/FischerL79} for Kripke structures and allows
to express fundamental properties in an easy and intuitive manner. PDL
for MSCs is closed under negation, it is a proper fragment of the
existential monadic second-order logic (EMSO) in terms of
expressiveness (but it is no syntactic fragment)
\cite{DBLP:journals/corr/abs-1007-4764}, and the logic $\text{TLC}^-$
considered by Peled \cite{DBLP:conf/forte/Peled00} is a fragment of
it.
PDL distinguishes between local and global formulas. The former ones are
evaluated at a specific event of an MSC whereas the latter are Boolean
combinations of local formulas quantifying existentially over all events of an
MSC. Consider for example the local formula $\alpha=p!q\wedge
\neg\pth{\proc^\ast}p?q$. An event satisfies~$\alpha$ if it is a send event of a
message from process $p$ to $q$ which is not followed by a reply message from
$q$ to~$p$. The global formula $\mathsf{E}\alpha$ expresses that there exists
such an event $v$.

By a rather involved translation of PDL formulas into CFMs, Bollig,
Kuske, and Meinecke demonstrated in
\cite{DBLP:journals/corr/abs-1007-4764} that the bounded model
checking problem for CFMs and PDL can be decided in polynomial space.
However, by means of this approach, Bollig et al.~were not able to
support the popular converse operator. The latter, introduced in
\cite{DBLP:conf/focs/Pratt76}, is an extension of PDL which allows to
walk back and forth within an MSC using a single path expression of
PDL. For example one can specify a path expression
$(\proc^{-1};\msg)^\ast$ describing ``zigzag-like'' paths going back
on a process and traversing a send event in an alternating manner. It
is an open question whether PDL formulas enriched with the converse
operator can be translated into CFMs. Bollig et al.~only managed to
provide an operator which enables path expressions to either walk
backward or forward.

In the present work, we consider CRPDL over MSCs which is PDL equipped
with the operators converse ($\_^{-1}$) and repeat ($\_^\omega$)
\cite{DBLP:conf/stoc/Streett81}. The latter allows to express that a
path expression can be repeated infinitely often. For example, an
event $v$ on process $p$ satisfies $\pth{\proc}^\omega$ if there are
infinitely many events on $p$ succeeding $v$.
We are able to demonstrate that the bounded model checking problem of
CFMs and CRPDL is in PSPACE and therefore generalize the model
checking result from \cite{DBLP:journals/corr/abs-1007-4764}. In order
to obtain this result, we define multi-way alternating parity automata
over MSCs which we call local message sequence chart automata (or
local MSCAs for short).
Local MSCAs are started at specific events of an MSC and accept sets
of pointed MSCs which are pairs of an MSC $M$ and an event~$v$
of~$M$. Using a game theoretic approach, it can be shown that local
MSCAs are closed under complementation. We demonstrate that every
local formula $\alpha$ of CRPDL can be translated in polynomial space
into an equivalent local MSCA whose size is linear in the size of
$\alpha$ --- this can be done independently from any channel bound. We
also define global MSCAs consisting of a local MSCA $\cM$ and a set of
global initial states. A global initial state is a tuple of states
$(\iota_1,\iota_2,\ldots,\iota_n)$ where $n$ is the number of
processes. If, for every process~$p$, there exists an accepting run of
$\cM$ starting in the minimal event of $p$ and the initial
state~$\iota_p$, then the global MSCA accepts the whole MSC. For every
global formula $\varphi$, we can construct in polynomial space a
global MSCA $\cG$ such that $\cG$ precisely accepts the set of models
of $\varphi$. After fixing a channel bound $B$, the automaton $\cG$ is
then transformed into a two-way alternating word automaton and, after
that, into a B\"uchi automaton recognizing the set of all $B$-bounded
linearizations of the models of $\varphi$.

In the literature, one can basically find two types of approaches to
turn a temporal formula into a Büchi automaton. On the one hand, Vardi
and others
\cite{DBLP:books/sp/cstoday95/Vardi95,DBLP:conf/mfcs/GastinO03}
transformed LTL formulas into alternating automata in one single step
and, afterwards, these alternating automata were translated into Büchi
automata. On the other hand, there were performed inductive
constructions which lead to a Büchi automaton without the need for an
intermediate step
\cite{DBLP:conf/icalp/KestenPR98,DBLP:journals/fuin/GastinK07,DBLP:journals/iandc/GastinK10,DBLP:journals/corr/abs-1007-4764}.
In the present work, we combine these two approaches to obtain a very
modular and easy to understand proof. For a given \mbox{CRPDL}
formula, we inductively construct an alternating automaton which is
later translated into a Büchi automaton. In this process, we utilize a
new concept called concatenation states.  These special states allow
the concatenation of local MSCAs.  For example, if $\cM$ is the local
MSCA obtained for the formula $\pth{\proc}\true$, then we can
concatenate two copies of $\cM$ to obtain an automaton for the formula
$\pth{\proc;\proc}\true$.

\paragraph{\textbf{Outline.}}
We proceed as follows. In Sect.~\ref{sec:preliminaries}, we define
MSCs, CRPDL, MSCAs, and give introductory examples. In
Sect.~\ref{sec:complementation}, we show that local MSCAs are
effectively closed under complementation. In
Sect.~\ref{sec:translationOfLocalFormulas}, we construct, for every
local CRPDL formula $\alpha$, a local MSCA which precisely accepts the
models of $\alpha$. In Sect.~\ref{sec:translationOfGlobalFormulas}, we
effectively show that, for every global CRPDL formula $\varphi$, the
set of models of $\varphi$ is the language of a global MSCA. In the
sections~\ref{sec:satisfiability} and \ref{sec:modelChecking}, we
prove that the bounded satisfiability problem for CRPDL and the
bounded model checking problem for CRPDL and CFMs both are
PSPACE-complete.

\bigskip

\noindent A conference version of this paper was published as
\cite{DBLP:conf/concur/Mennicke12}.

\paragraph{\textbf{Acknowledgements.}} The author likes to express his
sincere thanks to his doctoral adviser Dietrich Kuske for his guidance
and valuable advice. Furthermore, he is grateful to Benedikt Bollig
for comments leading to a considerable technical simplification. This
paper also greatly benefits from the detailed reviews and helpful
remarks of the anonymous referees.

\section{Preliminaries}\label{sec:preliminaries}

We let $\poly{n}$ denote the set of polynomial functions in one
argument. For every natural number $n\geq1$, we set
$[n]=\{1,2,\ldots,n\}$.

We fix a finite set $\mathbb{P}=\{1,2,\ldots,|\mathbb{P}|\}$ of processes. Let
$\mathsf{Ch}=\{(p,q)\in\mathbb{P}^2\mid p\neq q\}$ denote the set of
\emph{communication channels}. For all $p\in\mathbb{P}$, we define a local
alphabet $\Sigma_p=\{p!q,p?q\mid q\in\mathbb{P}\setminus\{p\}\}$ which we use
in the following way. An event labelled by $p!q$ marks the send event
of a message from process $p$ to process $q$ whereas $p?q$ is the
label of a receive event of a message sent from $q$ to $p$. We set
$\Sigma=\bigcup_{p\in\mathbb{P}}\Sigma_p$. Since $\mathbb{P}$ is finite, the
local alphabets $\Sigma_p$ and $\Sigma$ are also finite. 

\subsection{Message Sequence Charts}

Message sequence charts model the behavior of a finite set of parallel
processes communicating using FIFO channels. The following example
shows that they come with a convenient graphical representation.
\begin{figure}[t]
  \begin{center}
    \begin{tikzpicture} [auto,node distance=0.6cm,>=latex]
      \begin{scope}[rectangle,minimum width=0.55cm,font=\scriptsize]
        \node (p1) [draw] {$1$}; 
        \node (p2) [draw,right of=p1,xshift=1.55cm] {$2$};
      \end{scope}
    
      \foreach \i in {1,2} {
        \path (p\i.south) edge ([yshift=-3.7cm] p\i.south);
      }
    
      \foreach \proc/\number/\y [evaluate=\y as \yn using \y*0.6] in {
        1/1/0,1/2/1,1/3/2,1/4/3,1/5/4,1/6/5,
        2/1/0,2/2/1,2/3/2,2/4/3,2/5/4,2/6/5
      } {
        \node[draw,fill,circle,inner sep=1.5pt] (v\proc\number) [below of=p\proc,yshift=-\yn cm] {};
        \node[anchor=south west] at (v\proc\number) {};
      }
    
      \path[->]
      (v11) edge (v21)
      (v22) edge (v14)
      (v12) edge (v23)
      (v24) edge (v15)
      (v13) edge (v25)
      (v26) edge (v16);
      
      % \draw[left] (v11) node {\footnotesize $p!q$};
      % \draw[left] (v12) node {\footnotesize$p!q$};
      % \draw[left] (v13) node {\footnotesize$p!q$};
      % \draw[left] (v14) node {\footnotesize$p?q$};
      % \draw[left] (v15) node {\footnotesize$p?q$};
      % \draw[left] (v16) node {\footnotesize$p?q$};
        
      % \draw[right] (v21) node {\footnotesize$q?p$};
      % \draw[right] (v23) node {\footnotesize$q?p$};
      % \draw[right] (v25) node {\footnotesize$q?p$};
      % \draw[right] (v22) node {\footnotesize$q!p$};
      % \draw[right] (v24) node {\footnotesize$q!p$};
      % \draw[right] (v26) node {\footnotesize$q!p$};
    \end{tikzpicture}
  \end{center}
  %\caption{The example MSC $M$ from Example~\ref{ex:MSC}.}
  \caption{An example of a finite MSC.}
  \label{fig:MSC}
\end{figure}
\begin{exa}\label{ex:MSC}
  Figure~\ref{fig:MSC} shows a finite MSC $M$ over the set of
  processes $\mathbb{P}=\{1,2\}$.
\end{exa}
In the graphical representation of an MSC $M$ over the set of
processes $\mathbb{P}$, there is a vertical axis for every process from
$\mathbb{P}$. On the edge for process $p\in\mathbb{P}$, the events occurring on $p$
are drawn as small black circles. Thus, a linear ordering
$\preceq_p^M$ on the set of events from process $p$ is implicitly
defined. In the formal definition of an MSC, which we give in the
following, the direct successor relation induced by the linear
ordering $\preceq_p^M$ is given by $\proc_p^M$. For technical
convenience, we force processes to contain at least one
event. Messages sent between two processes are depicted by arrows
pointing from the send event to the matching receive event.  Formally,
messages are represented by the binary relation $\msg^M$ and, for
every send event, there exists a matching receive event and vice
versa.
\begin{defi}
  A \emph{message sequence chart} (\emph{MSC}) is a
  structure $$M=\big(V^M,(\proc_p^M)_{p\in\mathbb{P}},\msg^M,\lambda^M\big)$$
  where
  \begin{iteMize}{$\bullet$}
  \item $V^M$ is a set of \emph{events}, 
  \item $\proc_p^M,\msg^M\subseteq (V^M \times V^M)$ for all $p\in\mathbb{P}$,
  \item $\lambda^M\colon V^M \to \Sigma$ is a labeling function,
  \item for all $p\in\mathbb{P}$, the relation $\proc_p^M$ is the direct
    successor relation of a linear order $\preceq_p^M$ on
    $V_p^M:=\{v\in V^M\mid \lambda(v)\in\Sigma_p\}$,
  \item $(V_p^M,\preceq_p^M)$ is non-empty and finite or isomorphic to
    $(\mathbb{N},\leq)$,
  \item for all $v,w\in V^M$, we have $(v,w)\in\msg^M$ if and only
    if there exists $(p,q)\in\mathsf{Ch}$ such that
    $\lambda^M(v)=p!q$, $\lambda^M(w)=q?p$, and
    \[|\{u\in V^M \mid \lambda^M(u) = p!q, u \preceq_p^M v\}|=|\{u\in
    V^M \mid \lambda^M(u)=q?p,u\preceq_q^M w\}|\,,\]
  \item for every $v\in V^M$ there exists $w\in V^M$ such that $(v,w)\in\msg\cup\msgi$.
  \end{iteMize}
  If $v\in V^M$, then we denote by $P_M(v)$ the process at which $v$
  is located, i.e., $P_M(v)=p$ if and only if
  $\lambda^M(v)\in\Sigma_p$. Finally, if $v\in V^M$, then the pair
  $(M,v)$ is called a \emph{pointed MSC}.
\end{defi}
\begin{defi}
  We fix the set $\mathbb{M}=\{\proc,\proci,\msg,\msgi,\id\}$ of
  directions. An MSC~$M$ induces a partial function $\eta_M:(V^M\times
  V^M) \to \mathbb{M}$. For all $v,v'\in V^M$, we define
  \[
    \eta_M(v,v')=\begin{cases}
      \proc&\text{if $(v,v')\in\proc^M$}\\
      \proci&\text{if $(v',v)\in\proc^M$}\\
      \msg&\text{if $(v,v')\in\msg^M$}\\
      \msgi&\text{if $(v',v)\in\msg^M$}\\
      \id&\text{if $v=v'$}\\
      \text{undefined}&\text{otherwise}
    \end{cases}
  \]
  where $\proc^M=\bigcup_{p\in\mathbb{P}}\proc_p^M$.
\end{defi}

\subsection{Propositional Dynamic Logic with Converse and Repeat}

In this section, we introduce a new logic called propositional dynamic logic
with converse and repeat (or \mbox{CRPDL} for short). In CRPDL, we distinguish between
\emph{local} and \emph{global} formulas.  The former ones are evaluated at
specific events of an MSC. The latter are positive Boolean combinations of
properties of the form ``there exists an event satisfying a local formula'' or
``all events satisfy a local formula''.
\begin{defi}
  \emph{Local formulas} $\alpha$ and \emph{path expressions} $\pi$ of
  CRPDL are defined by the following grammar, where $D\in\mathbb{M}$
  and $\sigma$ ranges over the alphabet $\Sigma$:
  \begin{align*}
    \alpha&::=\sigma\mid\neg\alpha\mid\pth{\pi}\alpha\mid\pth{\pi}^\omega\\
     \pi&::=D\mid\{\alpha\}\mid\pi;\pi\mid\pi+\pi\mid\pi^\ast
  \end{align*}
  Formulas of the form $\pth{\pi}\alpha$ are called \emph{path
    formulas}. The \emph{size} of a local formula $\alpha$ is the
  length of the string $\alpha$.
\end{defi}
Note that $\proci$ and $\msgi$ form the converse operator
\cite{DBLP:conf/focs/Pratt76} which allows to walk back and forth
within an MSC using a single path expression.  The formula
$\pth{\pi}^\omega$ provides the functionality of the repeat operator
\cite{DBLP:conf/stoc/Streett81}. It allows to express that a path
expression can be repeated infinitely often.

Intuitively, a path formula $\pth{\pi}\alpha$ expresses that one can
move along a path described by $\pi$ and then $\alpha$ holds. In the
following formal definition of the semantics of local formulas, we
write $\reach_M(v,\pi)$ to denote the set of events which can be
reached from $v$ using a path described by $\pi$. A formal definition
of $\reach_M(v,\pi)$ is given at the end of
Definition~\ref{def:localFormulas}.
\begin{defi}\label{def:localFormulas}
  Let $(M,v)$ be a pointed MSC, $\sigma\in\Sigma$, $D\in\mathbb{M}$, $\alpha$ be a local
  formula, $\pi,\pi_1,\pi_2$ be path expressions. We define:
  \begin{align*}
    M,v\models\sigma&
      \iff\lambda^M(v)=\sigma\\
    M,v\models\neg\alpha&
      \iff M,v\not\models\alpha\\
    M,v\models\pth{D}\alpha&
      \iff\text{there exists }v'\text{ with }\eta_M(v,v')=D
      \text{ and }M,v'\models\alpha\\
    M,v\models\pth{\{\alpha\}}\beta&
      \iff M,v\models\alpha\text{ and }M,v\models\beta\\
    M,v\models\pth{\pi_1+\pi_2}\alpha&
      \iff M,v\models\pth{\pi_1}\alpha\text{ or }M,v\models\pth{\pi_2}\alpha\\
    M,v\models\pth{\pi_1;\pi_2}\alpha&
      \iff M,v\models\pth{\pi_1}\pth{\pi_2}\alpha\\
    M,v\models\pth{\pi^\ast}\alpha&
      \iff\text{there exists an }n\geq0\text{ with }
      M,v\models(\pth{\pi})^n\alpha\\
    M,v\models\pth{\pi}^\omega&
      \iff\text{there exist infinitely many events $v_0,v_1,\ldots$ such that}\\
    &\hspace{2cm}\text{ $v_0=v$ and $v_{i+1}\in\reach_M(v_i,\pi)$ for all $i\geq0$}
  \end{align*}
  where $\reach_M(v,\pi)$ is inductively defined as follows:
  \begin{align*}
    \reach_M(v,D)&=\begin{cases}\{v'\}&\text{if $\eta_M(v,v')=D$}\\
      \emptyset&\text{otherwise}\end{cases}\\
    \reach_M(v,\{\alpha\})&=\begin{cases}\{v\}&\text{if
      }M,v\models\alpha\\\emptyset&\text{otherwise}\end{cases}\\
    \reach_M(v,\pi_1;\pi_2)&=\textstyle\bigcup_{v'\in \reach_M(v,\pi_1)}\reach_M(v',\pi_2)\\
    \reach_M(v,\pi_1+\pi_2)&=\reach_M(v,\pi_1)\cup\reach_M(v,\pi_2)\\
    \reach_M(v,\pi^\ast)&=\{v\}\cup\textstyle\bigcup_{n\geq1}\reach_M(v,\pi^n)
  \end{align*}
  By $\lang(\alpha)$ we denote the set of pointed MSCs which satisfy $\alpha$.
\end{defi}
We set $\true=\sigma\vee\neg\sigma$ for some $\sigma\in\Sigma$. If
$\alpha=\pth{\pi}\true$, then we define
\[\reach_M(v,\alpha)=\reach_M(v,\pi)\,\text{.}\]
Furthermore, we use $\alpha_1\wedge\alpha_2$ as an abbreviation for
$\pth{\{\alpha_1\}}\alpha_2$ and write $\alpha_1\vee\alpha_2$ for the
formula $\neg(\neg\alpha_1\wedge\neg\alpha_2)$. Finally, for all $q
\in \mathbb{P}$, we define $P_q=\bigvee_{p\in\mathbb{P},q\neq p}(q!p\vee q?p)$. For
every pointed MSC $(M,v)$, we have $M,v\models P_q$ if and only if
$P_M(v)=q$.
\begin{rem}\label{remark:pathFormulas}
  It can be easily seen that $M,v\models\pth{\pi}\alpha$ if and only
  if $M,v\models\pth{\pi;\{\alpha\}}\true$. Because of this fact,
  every time we are dealing with path formulas in the future, we will
  assume that $\alpha=\true$.
\end{rem}
\begin{exa}
  The existential until construct $\alpha\mathsf{EU}\beta$
  \cite{DBLP:journals/tcs/KatzP90} can be expressed by the local formula
  $\pth{(\{\alpha\};(\proc+\msg))^\ast}\beta$.
\end{exa}
We now define global formulas which are positive Boolean combinations
of properties of the form ``there exists an event satisfying a local
formula $\alpha$'' or ``all events satisfy a local formula $\alpha$''.
\begin{defi}
  The syntax of \emph{global formulas} is given by the grammar
  \[
    \varphi::=\mathsf{E}\alpha\mid\mathsf{A}\alpha\mid\varphi\vee\varphi
    \mid\varphi\wedge\varphi
  \]
  where $\alpha$ ranges over the set of local formulas. Their semantics is as
  follows: If $M$ is an MSC, $\alpha$ is a local formula, and
  $\varphi_1,\varphi_2$ are global formulas, then
  \begin{align*}
    M\models\mathsf{E}\alpha&\iff\text{there exists $v\in V^M$ with $M,v\models\alpha$}\,\text{,}\\
    M\models\mathsf{A}\alpha&\iff M,v\models\alpha\text{ for all }v\in V^M \,\text{,}\\
    M\models\varphi_1\vee\varphi_2&\iff M\models\varphi_1\text{ or
    }M\models\varphi_2 \,\text{, and}\\
    M\models\varphi_1\wedge\varphi_2&\iff M\models\varphi_1\text{ and
    }M\models\varphi_2 \,\text{.}
  \end{align*}
  We define the \emph{size} of a global formula $\varphi$ to be the
  length of the string $\varphi$. By $\lang(\varphi)$, we denote the
  set of MSCs $M$ with $M\models\varphi$.
\end{defi}
Note that even though there are no negation operators allowed in global
formulas, the expressible properties are still closed under negation. This is
because conjunction and disjunction operators as well as existential and
universal quantification are available. 
\begin{exa}[\cite{DBLP:journals/corr/abs-1007-4764}]\label{ex:CRPDL}
  Let $\beta_{p}=\pth{\proc^\ast;\msg;\proc^\ast;\msg}P_p$. If $(M,v)$
  is a pointed MSC such that $M,v\models\beta_p$, then process $p$ can
  be reached from $v$ with exactly two messages. If $M$ is the MSC
  from Fig.~\ref{fig:MSC}, then $M,v\models\beta_1$ if and only if $v$
  is one of the first three events on process $1$. The global formula
  $\varphi_p=\mathsf{A}\beta_p$ states that $\beta_p$ holds for every
  event of an MSC $M$ (which in particular implies that $M$ is
  infinite).
\end{exa}
\begin{exa}\label{ex:CRPDL2}
  An MSC $M$ satisfies
  $\mathsf{E}\bigwedge_{p\in\mathbb{P}}(\pth{(\proc+\msg+\proci+\msgi)^\ast}P_p)$
  if and only if the graph $(V^M, \proc^M \cup \msg^M \cup
  (\proc^M)^{-1} \cup (\msg^M)^{-1})$ is connected.
\end{exa}
\begin{exa}
  Now, let $\pi_p=((\proc+\msg)^\ast;\{P_p\})$ for every
  $p\in\mathbb{P}$.  Imagine that $M$ is an MSC which models the
  circulation of a single token granting access to a shared
  resource. Then
  $M\models\mathsf{E}\pth{\pi_1;\pi_2;\ldots;\pi_{|\mathbb{P}|}}^\omega$ if
  and only if no process ever gets excluded from using the shared
  resource.
\end{exa}

\subsection{Message Sequence Chart Automata (MSCA)}

In this section, we give the definition of MSCAs which basically are
multi-way alternating parity automata walking forth and back on the
process and message edges of MSCs.
We first define local MSCAs which are started at individual events of
an MSC.  They also come with a so called concatenation state. This
type of state is used to concatenate local MSCAs in order to obtain
more complex local MSCAs. Using this technique, we will show in a subsequent
section that every local formula of CRPDL can be transformed into a
local MSCA.
\begin{defi}
  If $X$ is a non-empty set, then $\cB^+(X)$ denotes the \emph{set of
    all positive Boolean expressions over $X$} together with the
  expression $\bot$.  The latter expression is always evaluated to
  false. We say that $Y\subseteq X$ is a \emph{model} of
  $E\in\cB^+(X)$ and write $Y\models E$ if~$E$ is evaluated to true
  when assigning true to every element contained in $Y$ and assigning
  false to all other elements from $X\setminus Y$.
  The set $Y\subseteq X$ is a \emph{minimal} model of $E$ if $Y\models
  E$ and $Z\not\models E$ for all $Z\subsetneq Y$.
  We denote the set of all models of $E$ by $\model(E)$ whereas we
  write~$\minmodel{E}$ for the set of all minimal models of $E$.
\end{defi}
For instance, $\{a,b,c\}$, $\{a,b\}$, $\{a,c\}$, $\{b,c\}$, and
$\{a\}$ are all models of the positive Boolean expression
$a\lor(b\land c)\in\cB^+(\{a,b,c\})$. However, only $\{a\}$, and
$\{b,c\}$ are minimal models.
\begin{defi}
  A \emph{local message sequence chart automaton (local MSCA)} is a
  quintuple $\cM=(S,\delta,\iota,\cstate,\rank)$ where
  \begin{iteMize}{$\bullet$}
    \item $S$ is a finite set of states,
    \item $\delta\colon (S\times\Sigma)\to\cB^+(\mathbb{M}\times
    S)$ is a transition function,
    \item $\iota\in S$ is an initial state,
    \item $\cstate\in S$ is a concatenation state, and
    \item $\rank:S\to\mathbb{N}$ is a ranking function.
  \end{iteMize}
  The \emph{size} of $\cM$ is $|S| + |\delta|$. If we do not pay
  attention to the concatenation state $\cstate$, then we sometimes
  write $(S,\delta,\iota,\rank)$ instead of
  $(S,\delta,\iota,\cstate,\rank)$.  If $s\in S$, $\sigma\in\Sigma$,
  and $\tau\in\minmodel{\delta(s,\sigma)}$, then $\tau$ is called a
  \emph{transition}.
\end{defi}
For example, the transition $\tau = \{(\proc, s_1), (\msg, s_2)\}$
which is a minimal model of the expression $(\proc, s_1) \land ((\msg
, s_2) \lor (\msgi, s_2))$ can be interpreted in the following way:
Let us assume that $\cM$ is in state $s\in S$ at an event $v$. If it
performs the transition $\tau$, then it changes, in parallel, from the
state $s$ into the states $s_1$ and $s_2$, i.e., the run splits. In
the case of state $s_1$, it moves to the event succeeding the event
$v$ on the current process. For $s_2$, the automaton walks along a
message edge to the receive event of the message sent in $v$.
Hence, the conjunctive connectives implement \emph{universal
  branching} whereas the disjunctive connectives realize
\emph{existential branching} and nondeterminism, respectively.  As a
consequence, local MSCAs are alternating automata and their runs may
split. Therefore, in order to be able to define runs of local MSCAs,
we first introduce labelled trees.

Later, in the construction of local MSCAs from local formulas, the
concatenation state~$\cstate$ of $\cM$ will be used to concatenate
local MSCAs for simple local formulas in order to obtain automata
which are equivalent to more complex formulas.
\begin{exa}\label{ex:MSCA}
  Let $p\in\mathbb{P}$ be fixed. Consider the local MSCA $\cM = (S, \delta,
  s_1, \rank)$ which is depicted in
  Fig.~\ref{fig:exampleLocalMSCA}. Its set of states $S$ consists of
  the three states $s_1$, $s_2$, and $s_3$ where $s_1$ is the initial
  state. Each state is depicted by a circle. The label of the circles
  also tells us the rank of each state. For example, $s_1\mid 1$
  expresses that $\rank(s_1) = 1$. Furthermore, we have $\rank(s_2)=1$ and
  $\rank(s_3)=0$. Transitions are depicted by arrows. For instance,
  the arrow from $s_1$ to $s_2$ labelled by $\Sigma$ and $\msg$ says
  that the automaton can make a transition from $s_1$ to $s_2$ by
  following a message edge and going to the matching receive event,
  respectively. We write $\Sigma$ because this transition can be
  executed no matter what the label of the current event
  is. Alternatively, the automaton can stay in state $s_1$ by going to
  the successor of the current event --- this is expressed by the loop
  at $s_1$. More formally, for all $\sigma \in \Sigma$, we have
  $\delta(s_1,\sigma)=(\proc,s_1)\lor(\msg,s_2)$, $\delta(s_3,\sigma)
  = \bot$, and
  \[\delta(s_2,\sigma)=\begin{cases}
    (\id,s_3)&\text{if $\sigma=q!p$ where $q\in\mathbb{P}\setminus\{p\}$}\\
    (\proc,s_2)&\text{otherwise.}
  \end{cases}\]
\end{exa}
\begin{figure}
  \begin{center}\small
    \pgfdeclarelayer{background}
    \pgfdeclarelayer{foreground}
    \pgfsetlayers{background,main,foreground}
    \begin{tikzpicture}[>=stealth',semithick,,shorten >=1pt,shorten <=1pt,
      on grid,node distance=3cm,initial text=,initial distance=0.7cm,
      every state/.style={fill=white,draw=black,text=black,
      minimum size=1cm,inner sep=1pt}]
      
      \begin{scope}
        \node[state,initial] (s1) {$s_1\mid1$};
        \node[state,right of=s1] (s2) {$s_2\mid 1$};
        \node[state,right of=s2,xshift=1.5cm] (s3) {$s_3\mid 0$};
      \end{scope}
      
      \path[->] (s1) edge (s2);
      \path[->] (s1) edge [loop above] node {$\Sigma$, $\proc$} ();
      \path[->] (s2) edge (s3);
      \path[->] (s2) edge [loop above] node {$\Sigma$, $\proc$} ();
      
      \node[anchor=south west] at ([xshift=0.4cm,yshift=0.1cm] s1)
      {$\Sigma$};

      \node[anchor=south east] at ([xshift=-0.4444cm,yshift=0.1cm] s2)
      {$\msg$};

      \node[anchor=south west] at ([xshift=0.4cm,yshift=0.1cm] s2)
      {$\{q!p\mid q\neq p\}$};

      \node[anchor=south east] at ([xshift=-0.4444cm,yshift=0.1cm] s3)
      {$\id$};
      
%      \begin{pgfonlayer}{background}
%        \pgfsetcornersarced{\pgfpoint{5mm}{5mm}}
%        \fill[draw=gray,fill=white!90!black] ([yshift=40] s1) rectangle
%        ([yshift=-40] s3);
%      \end{pgfonlayer}
    \end{tikzpicture}
  \end{center}
  \caption{The local MSCA $\cM$ from Example \ref{ex:MSCA}.}
  \label{fig:exampleLocalMSCA}
\end{figure}
Note that the above example makes use of existential branching only
whereas the MSCA of the next example also implements universal
branching.
\begin{exa}\label{ex:MSCA2}
  Consider the local MSCA from Fig.~\ref{fig:exampleLocalMSCA2}. Note
  that universal branching is depicted by forked arrows. We have
  $\cM'=(S\cup\{t_1,t_2\},\delta\cup\delta',t_1,\rank\cup\rank')$
  where $\cM=(S,\delta,s_1,\rank)$ is the local MSCA from
  Example~\ref{ex:MSCA}, $\rank'(t_1)=1$, $\rank'(t_2)=0$, and
  $\delta(t_1,\sigma)=(\id,t_2)\land(\id,s_1)$ and
  $\delta(t_2,\sigma)=(\proc,t_1)$ for all $\sigma\in\Sigma$.
\end{exa}
\begin{figure}
  \begin{center}\small
    \pgfdeclarelayer{background}
    \pgfdeclarelayer{foreground}
    \pgfsetlayers{background,main,foreground}
    \begin{tikzpicture}[>=stealth',semithick,,shorten >=1pt,shorten <=1pt,
      on grid,node distance=3cm,initial text=,initial distance=0.7cm,
      every state/.style={fill=white,draw=black,text=black,
      minimum size=1cm,inner sep=1pt}]
      
      % states
      \node[state,initial] (t1) {$t_1\mid 1$};
      \node[state,below of=t1] (t2) {$t_2\mid 0$};
      \node[state,right of=t1,yshift=-1.5cm] (s1) {$s_1\mid 1$};
      \node[state,right of=s1] (s2) {$s_2\mid 1$};
      \node[state,right of=s2,xshift=1.5cm] (s3) {$s_3\mid 0$};

      % helper node
      \node[outer sep=0pt,inner sep=0pt] (helper) at
      ([xshift=-2.3cm,yshift=0.5cm] s1) {};
      
      % transitions
      \path[->] (s1) edge (s2);
      \path[->] (s1) edge [loop above] node {$\Sigma$, $\proc$} ();
      \path[->] (s2) edge (s3);
      \path[->] (s2) edge [loop above] node {$\Sigma$, $\proc$} ();

      \path[shorten >=-1pt] (t1) edge (helper);

      \path[shorten <=-1pt,->] (helper) edge[bend right=20] (s1);
      \path[shorten <=-1pt,->] (helper) edge[bend left] (t2);

      \path[->] (t2) edge[bend left] (t1);

      % transition labels
      \node[anchor=south west] at ([xshift=0.5cm,yshift=0.1cm] s1) {$\Sigma$};
      \node[anchor=south east] at ([xshift=-0.5cm,yshift=0.1cm] s2) {$\msg$};
      \node[anchor=south west] at ([xshift=0.5cm,yshift=0.1cm] s2)
      {$\{q!p\mid q\neq p\}$};
      \node[anchor=south east] at ([xshift=-0.5cm,yshift=0.1cm] s3) {$\id$};

      \node[anchor=north west] at ([xshift=0.5cm,yshift=-0.2cm] t1) {$\Sigma$};
      \node[anchor=south east] at ([xshift=-0.5cm,yshift=0.1cm] s1) {$\id$};
      \node[anchor=south west] at ([xshift=0.5cm,yshift=0.1cm] t2) {$\id$};

      \node[anchor=north east] at ([xshift=-0.5cm,yshift=-0.2cm] t1) {$\proc$};
      \node[anchor=south east] at ([xshift=-0.5cm,yshift=0.1cm] t2) {$\Sigma$};
    \end{tikzpicture}
  \end{center}
  \caption{The local MSCA $\cM'$ from Example \ref{ex:MSCA2}.}
  \label{fig:exampleLocalMSCA2}
\end{figure}
\begin{defi}
  A \emph{tree} is a directed, connected, cycle-free graph $(C,E)$
  with the set of nodes $C$ and the set of edges $E$ such that there
  exists exactly one node with no incoming edges (which is called
  \emph{root}) and all other nodes have exactly one incoming edge.
\end{defi}
We now define so-called $S$-labelled trees over pointed MSCs where $S$
is an arbitrary set. Later, the set $S$ will be the set of states of a
local MSCA.
\begin{defi}\label{def:labelledTree}
  Let $S$ be an arbitrary set, $M$ be an MSC, and $v\in M$. An
  \emph{$S$-labelled tree over $(M,v)$} is a quintuple
  $\rho=(C,E,r,\mu,\nu)$ where
  \begin{enumerate}[(1)]
    \item $(C,E)$ is a tree with root $r$,
    \item $\mu:C\to S$ is a labeling function,
    \item $\nu:C\to V^M$ is a positioning function with $\nu(r)=v$,
      % we require the following property because we want to state
      % Lemma 3.4 (for
      % every accepting run there exists exactly one configuration
      % labelled by a
      % concatenation state)
    \item $\mu(y_1)\neq\mu(y_2)$ or $\nu(y_1)\neq\nu(y_2)$ for all
      $(x,y_1),(x,y_2)\in E$ with $y_1\neq y_2$, and
    \item $\eta_M(\nu(x),\nu(y))$ is defined for all $(x,y)\in E$.
  \end{enumerate}
  The elements of $C$ are called \emph{configurations}. If $x\in C$,
  then $E_\rho(x)=\{y\in C\mid (x,y)\in E\}$ denotes the set of the
  direct successor configurations of $x$ in $\rho$.  For convenience,
  we identify~$\mu$ with its natural extension, i.e.,
  $\mu(x_1x_2x_3\ldots)=\mu(x_1)\mu(x_2)\mu(x_3)\ldots\in S^\ast\cup
  S^\omega$.
\end{defi}
We use $S$-labelled trees to define runs of local MSCAs.  The
condition (4) has no influence on the expressiveness of local MSCAs
but simplifies the proofs in
Section~\ref{sec:translationOfLocalFormulas}. Intuitively, it prevents
a local MSCA from doing unnecessary work. By item (5), we ensure that
an MSCA cannot jump within an MSC but must move along process or
message edges.
\begin{defi}
  Let $S$ be a set, $(M,v)$ be a pointed MSC, and
  $\rho=(C,E,r,\mu,\nu)$ be an $S$-labelled tree over $(M,v)$. A
  \emph{path} in $\rho$ of length $n\in\N\cup\{\omega\}$ is a sequence
  $x_1x_2x_3\ldots\in C^n$ such that $x_{i+1}\in E_\rho(x_i)$ for all
  $1\leq i<n$. It is a \emph{branch of $\rho$} if $x_1=r$ and
  $E_\rho(x_n)=\emptyset$ (provided that $n\in\N$).
\end{defi}
That means every branch of $\rho$ begins in the root of $\rho$ and
either leads to some leaf of $\rho$ or is infinite.
\begin{defi}
  If $C'\subseteq C$ such that $(C',E\cap(C'\times C'))$ is a tree
  with root $r'$, we denote by $\rho\restriction C'$ the
  \emph{restriction of $\rho$ to $C'$}, i.e., the $S$-labelled tree
  $(C',E',r',\mu',\nu')$ where $E'=E\cap(C'\times C')$,
  $\mu'=\mu\restriction C'$, and $\nu'=\nu\restriction C'$.
\end{defi}
We want the runs of local MSCAs to be maximal. That means that, during
a run, a local MSCA is forced to execute a transition if it is able to
do so. If the MSCA is unable to proceed, we say that it is stuck.
\begin{defi}\label{def:stuck}
  Let $M$ be an MSC and $\cM=(S,\delta,\iota,\rank)$ be a local
  MSCA. The automaton~$\cM$ is \emph{stuck} at $v\in V^M$ in the state
  $s\in S$ if for every transition
  $\tau\in\minmodel{\delta(s,\lambda^M(v))}$ there exists a movement
  $(D,s')\in\tau$ such that there exists no event $v'\in V^M$ with
  $\eta_M(v,v')=D$.
\end{defi}
We are now prepared to define runs of local MSCAs.
\begin{defi}\label{def:run}
  Let $\cM=(S,\delta,\iota,\rank)$ be a local MSCA and
  $\rho=(C,E,r,\mu,\nu)$ be an $S$-labelled tree over a pointed MSC
  $(M,v)$.  We define $\mathsf{tr}_\rho:C\to 2^{\mathbb{M}\times S}$
  to be the function which maps every $x\in C$ to the set
  \[\left\{ \big( \eta_M( \nu(x), \nu(x')), \mu(x') \big) \mid x' \in
    E_\rho(x) \right\}\,.\]
  The tree $\rho$ is a \emph{run} of $\cM$ on $(M,v)$ if
  $\mu(r)=\iota$ and, for all $x\in C$, the \emph{run condition} is
  fulfilled, i.e.,
  \begin{iteMize}{$\bullet$}
  \item if $E_\rho(x)\neq\emptyset$, then
    $\mathsf{tr}_\rho(x)\in\minmodel{\delta(\mu(x),\lambda^M(\nu(x)))}$,
    and
  \item if $E_\rho(x)=\emptyset$, then $\cM$ is stuck at the event
    $\nu(x)$ in state $\mu(x)$.
  \end{iteMize}
\end{defi}
\begin{defi}\label{def:accepting}
  Let $(s_i)_{i\geq1}\in S^\ast\cup S^\omega$ be a sequence of
  states. By $\mathsf{inf}((s_i)_{i\geq1})$, we denote the set of
  states occurring infinitely often in $(s_i)_{i\geq1}$. If
  $(s_i)_{i\geq1}$ is finite, then it is \emph{accepting} if it ends
  in a state $s$ whose rank $\rank(s)$ is even.  If it is infinite, it
  is accepting if the minimum of the ranks of all states occurring
  infinitely often is even, i.e., $\min\{\rank(s)\mid
  s\in\mathsf{inf}((s_i)_{i\geq1})\}$ is even.
  
  If $\rho$ is a run of $\cM$, and $\branch$ is a branch of $\rho$, then
  $\branch$ is \emph{accepting} if its label $\mu(\branch)$ is
  accepting. A run $\rho$ of $\cM$ is \emph{accepting} if every branch
  of $\rho$ is accepting.
  By $\lang(\cM)$, we denote the set of all pointed MSCs $(M,v)$ for
  which there exists an accepting run of $\cM$. Furthermore, for
  all~$p \in \mathbb{P}$, $\lang_p(\cM)$ is the set of MSCs $M$ with
  $(M,v)\in\lang(\cM)$ where $v$ is the minimal element from $V_p^M$
  with respect to $\preceq^M_p$.
\end{defi}
\begin{exa}
  Let $\cM$ and $\cM'$ be the MSCAs from the Examples \ref{ex:MSCA}
  and \ref{ex:MSCA2}, respectively. It can be easily checked that,
  for every pointed MSC $(M,v)$, we have $(M,v)\in\lang(\cM)$ if and
  only if $M,v\models\beta_p$ where
  $\beta_p=\pth{\proc^\ast;\msg;\proc^\ast;\msg}P_p$ is the formula
  from Example~\ref{ex:CRPDL}.
  In contrast, a pointed MSC $(M,v)$ is accepted by $\cM'$ if and only
  if $M,v'\models\beta_p$ for all $v'\in V^M$ with $v\preceq_p^M v'$.
\end{exa}
We also introduce the notion of global MSCAs which come as a local
MSCA together with a set of global initial states.
\begin{defi}
  A \emph{global message sequence chart automaton} (global MSCA) is a
  tuple $\cG=(\cM,I)$ where $\cM=(S,\delta,\iota,\rank)$ is a local
  MSCA and $I\subseteq S^{|\mathbb{P}|}$ is a set of global initial states.
  The language of $\cG$ is defined
  by
  \[\lang(\cG)=\bigcup_{(s_1,\ldots,s_{|\mathbb{P}|})\in I }\bigcap_{p\in\mathbb{P}}
    \lang_p(S,\delta,s_p,\rank)\,.\]
  The size of $\cG$ is the size of $\cM$.
\end{defi}
Intuitively, an MSC $M$ is accepted by $\cG$ if and only if there exists a
global initial state $(s_1,s_2,\ldots,s_{|\mathbb{P}|})\in I$ such that, for every
$p\in\mathbb{P}$, the local MSCA $\cM$ accepts $(M,v_p)$ when started in the
state $s_p$ where $v_p$ is the minimal event on process $p$.
\begin{exa}
  Let $\cG=(\cM',\{(t_1,\ldots,t_1)\})$ be the global MSCA where
  $\cM'$ is the local MSCA from Example~\ref{ex:MSCA2}. We have
  $M\in\lang(\cG)$ if and only if $M\models\varphi_p$ where
  $\varphi_p$ is the global formula from Example~\ref{ex:CRPDL}.
\end{exa}

\section{Closure under Complementation}
\label{sec:complementation}

If $\cM$ is a local MSCA, then a local MSCA $\cM^\#$ recognizing the
complement of $\lang(\cM)$ can be easily obtained. Basically one just
needs to exchange $\land$ and $\lor$ in the image of the transition
function of $\cM$ and update the ranking function.

To make this more precise, let us first define the dual expression
$\widetilde{E}$ of a positive Boolean expression $E$.
\begin{defi}
  Let $X$ be a set and $E\in\cB^+(X)$. Then the \emph{dual expression
  $\widetilde{E}$ of $E$} denotes the positive Boolean expression obtained by
  exchanging $\land$ and $\lor$ in $E$.
\end{defi}
Let us state the following two easy lemmas on positive Boolean expressions and
their dual counterparts.
\begin{lem}\label{lemma:dualExpression}
  Let $X$ be a set and $E\in\cB^+(X)$. Then, for all $Y\in\model(E)$ and
  $Z\in\model(\widetilde{E})$, we have $Y\cap Z\neq\emptyset$.
\end{lem}
\begin{proof}
  If $E=a$ for some $a\in X$, then the lemma easily follows. For the
  induction step, let us assume that $E=E_1\land E_2$ such that, for
  all $i\in[2]$, $Y\in\model(E_i)$, and $Z\in\model(\widetilde{E}_i)$,
  we have $Y\cap Z\neq\emptyset$. If $Y\in\model(E)$ and
  $Z\in\model(\widetilde{E})$, then, without loss of generality,
  $Y\models E_1$ and $Z\models\widetilde{E}_1$. It follows from the
  induction hypothesis that $Y\cap Z\neq\emptyset$.
  The case $E=E_1\lor E_2$ is shown analogously.
\end{proof}
\begin{lem}\label{lemma:dualTransition}
  Let $X$ be a set and $E\in\cB^+(X)$. If $Z\subseteq X$ such that $Z\cap
  Y\neq\emptyset$ for all $Y\in\model(E)$, then
  $Z\in\model(\widetilde{E})$.
\end{lem}
\begin{proof}
  If $E=a$ for some $a\in X$, then the lemma easily follows. For the
  induction step, let $E_1,E_2\in\cB^+(X)$ such that, for
  all $i\in[2]$, the following holds: if $Z\subseteq X$ and $Z\cap
  Y\neq\emptyset$ for all $Y\in\model(E_i)$, then
  $Z\in\model(\widetilde{E}_i)$.

  For the case $E=E_1\lor E_2$, let $Z\subseteq X$
  such that $Z\cap Y\neq\emptyset$ for all $Y\in\model(E)$. If
  $i\in[2]$ and $Y\in\model(E_i)$, then $Y\models E$. Hence, $Z\cap
  Y\neq\emptyset$ for all $i\in[2]$ and $Y\in\model(E_i)$. From our
  induction hypothesis it follows that $Z\models\widetilde{E}_1$ and
  $Z\models\widetilde{E}_2$ and, therefore, $Z\models\widetilde{E}$.
  Now, let us consider the case $E=E_1\land E_2$. Towards a
  contradiction, suppose that there exists a $Z\subseteq X$ such that
  $Z\cap Y\neq\emptyset$ for all $Y\in\model(E)$ and
  $Z\not\models\widetilde{E}$. Since
  $\widetilde{E}=\widetilde{E}_1\lor\widetilde{E}_2$, we have
  $Z\not\models\widetilde{E}_1$ and
  $Z\not\models\widetilde{E}_2$. From our induction hypothesis it
  follows that there exist $Y_1\in\model(E_1)$ and $Y_2\in\model(E_2)$
  with $Z\cap Y_1=Z\cap Y_2=\emptyset$. Since we also have $Y_1\cup
  Y_2\models E$, this is a contradiction to our definition of $Z$.
\end{proof}
We are now prepared to dualize local MSCAs.
\begin{defi}\label{def:dual_msca}
  Let $\cM=(S,\delta,\iota,\cstate,\rank)$ be a local MSCA.  The
  \emph{dual MSCA} $\cM^\#$ is the local MSCA
  $(S,\delta^\#,\iota,\cstate,\rank^\#)$ where
  \begin{iteMize}{$\bullet$}
    \item $\rank^\#(s)=\rank(s)+1$ for all $s\in S$ and
    \item $\delta^\#(s,\sigma)=\widetilde{\delta(s,\sigma)}$ for all
      $s\in S$ and $\sigma\in\Sigma$.
  \end{iteMize}
\end{defi}
\begin{rem}\label{remark:stateSequenceAcceptance}
  Let $\cM=(S,\Delta,\iota,\rank)$ be a local MSCA and
  $(s_i)_{i\geq1}\in S^\infty$ be a sequence of states. Because of our
  definition of $\rank^\#$, a state $s\in S$ has an even rank in $\cM$
  if and only if it has an odd rank in $\cM^\#$.  It follows that
  $(s_i)_{i\geq1}$ is accepting in $\cM$ if and only if it is not
  accepting in $\cM^\#$.
\end{rem}
If $(M,v)$ is a pointed MSC, $\rho$ is a run of $\cM$ on $(M,v)$, and
$\rho^\#$ is a run of $\cM^\#$ on $(M,v)$, then one can observe that
$\rho$ contains a branch $x_1x_2x_3\ldots$ and $\rho^\#$ contains a
branch $x_1'x_2'x_3'\ldots$ such that
$\mu(x_1x_2x_3\ldots)=\mu(x_1'x_2'x_3'\ldots)$, i.e. they are labelled
by the same sequence of states. Because of the fact stated in
Remark~\ref{remark:stateSequenceAcceptance}, $x_1x_2x_3\ldots$ is
accepting in~$\cM$ if and only if $x_1'x_2'x_3'\ldots$ is not
accepting in~$\cM^\#$.
By means of this observation, a result on parity games, and the ideas
presented in \cite{DBLP:journals/tcs/MullerS87}, we prove the
following theorem:
\begin{thm}\label{theorem:closedUnderComplement}
  If $\cM$ is a local MSCA and $(M,v)$ is a pointed MSC, then
  $$(M,v)\in\lang(\cM)\iff(M,v)\notin\lang(\cM^\#)\,\text{.}$$
\end{thm}
The rest of this section prepares the proof of the above
theorem. The actual proof can be found on page
\pageref{proof:closedUnderComplement}.
\begin{defi}
  Let $\cM=(S,\delta,\iota,\cstate,\rank)$ be a local MSCA and $(M,v)$
  be a pointed MSC. With $\cM$ and the pointed MSC $(M,v)$, we
  associate a \emph{game} $G(\cM,M,v)$ played by the two players
  \emph{Automaton} and \emph{Pathfinder} in the arena
  $(C_A,C_P,E_A,E_P)$ where $C_A = V^M \times S$, $C_P = V^M \times
  2^{\mathbb{M}\times S}$, $E_A\subseteq C_A\times C_P$, $E_P\subseteq
  C_P\times C_A$,
  \begin{multline*}
    \big((v,s),(v,\tau)\big)\in E_A\iff\tau\in\minmodel{\delta(s,\lambda^M(v))}
    \text{ and, for all $(D,s')\in\tau$, there exists}\\
    \text{an event $v'\in V^M$ such that $\eta_M(v,v')=D$}\,\text{,}
  \end{multline*}
  and
  \[\big((v,\tau),(v',s)\big)\in E_P\iff\text{there exists
    $D\in\mathbb{M}$ such that $(D,s)\in\tau$ and
  $\eta_M(v,v')=D$}\,\text{.}\]
  $C_A$ is the set of \emph{game positions} of the player Automaton.
  Analogously, at a position from $C_P$ it is Pathfinder's turn. The game
  position $(v,\iota)$ is called the \emph{initial position}.
\end{defi}
A \emph{play} of $G(\cM,M,v)$ starts at the initial position
$(v,\iota)$ from the set $C_A$, i.e., the player Automaton has to move
first. He chooses a transition $\tau$ from
$\minmodel{\delta(\iota,\lambda^M(v))}$ resulting in a game position
$(v,\tau)\in C_P$. Now, it is Pathfinder's turn who has to pick a
movement $(D,s)$ from $\tau$. This leads to the game position
$(v',s)\in C_A$ with $\eta_M(v,v')=D$. After that, Automaton has to
move next and so on. More formally, we define:
\begin{defi}
  Let $\cM=(S,\delta,\iota,\rank)$ be a local MSCA and $(M,v)$ be a
  pointed MSC. A \emph{partial play} $\xi$ of $G(\cM,M,v)$ is a
  sequence of one of the following two forms:
  \begin{enumerate}[(1)]
  \item $\xi=\big((v_i,s_i)(v_i,\tau_i)\big)_{1\leq i\leq n}\in
    (C_A C_P)^n$ where
    \begin{iteMize}{$\bullet$}
    \item $n\geq 1$
    \item $(v_1,s_1)=(v,\iota)$
    \item $((v_i,s_i),(v_i,\tau_i))\in E_A$ for all $1\leq i\leq n$
    \item $((v_i,\tau_i),(v_{i+1},s_{i+1}))\in E_P$ for all $1\leq
      i<n$
    \end{iteMize}
  \item $\xi=\big((v_i,s_i)(v_i,\tau_i)\big) _{1\leq i< n}
    (v_n,s_n)\in (C_P C_A)^{n-1} C_A$ where
    \begin{iteMize}{$\bullet$}
    \item $n\geq 1$
    \item $(v_1,s_1)=(v,\iota)$
    \item $((v_i,s_i),(v_i,\tau_i))\in E_A$ for all $1\leq i<n$
    \item $((v_i,\tau_i),(v_{i+1},s_{i+1}))\in E_P$ for all $1\leq
      i<n$
    \end{iteMize}
  \end{enumerate}
  The sequence $(s_i)_{1\leq i\leq n}\in S^n$ is called the \emph{label}
  of $\xi$. By $\xi\restriction C_A$ we denote the sequence
  $(v_1,s_1)(v_2,s_2)\ldots$ which is obtained by restricting $\xi$ to
  the positions from $C_A$.

  The sequence
  $\xi=(v_1,s_1)\big((v_i,\tau_i),(v_{i+1},s_{i+1})\big)_{1\leq i< n}$
  with $n\in(\mathbb{N}\setminus\{0\})\cup\{\infty\}$ is a \emph{play}
  of $G(\cM,M,v)$ if the following conditions are fulfilled:
  \begin{iteMize}{$\bullet$}
  \item $\big((v_i,s_i)(v_i,\tau_i)\big) _{1\leq i\leq j}
    (v_{j+1},s_{j+1})\in (C_A C_P)^j C_A$ is a partial
    play for all $0\leq j<n$
  \item if $n\in\mathbb{N}$, then there does not exit a $(v,\tau)\in
    C_P$ such that $\big((v_n,s_n),(v,\tau)\big)\in E_A$, i.e., the
    player Automaton cannot move any more.
  \end{iteMize}
  If $\xi$ is a play, then the label of $\xi$ is the sequence
  $s_1s_2s_3\ldots\in S^\infty$. Automaton is declared the
  \emph{winner} of the play $\xi$ if the label of $\xi$ is accepting
  in $\cM$. Otherwise, the play is won by Pathfinder.
\end{defi}
We define (memoryless) (winning) strategies in the usual way:
\begin{defi}
  A strategy of player Automaton in the game $G(\cM,M,v)$ is a total
  function $f:((C_A C_P)^\ast C_A) \to C_P$. A (partial)
  play $\xi=(v_1,s_1)\big((v_i,\tau_i)(v_{i+1},s_{i+1})\big)_{1\leq
    i<n}$ is called a \emph{(partial) $f$-play} if
  $(v_i,\tau_i)=f\big((v_1,s_1)(v_1,\tau_1)\ldots(v_i,s_i)\big)$ for
  every $1\leq i<n$. The strategy $f$ is called \emph{memoryless} if
  $f(\xi_1)=f(\xi_2)$ for all $(v,s)\in C_A$ and
  $\xi_1,\xi_2\in(C_A C_P)^\ast\{(v,s)\}$.  Furthermore, $f$ is
  a \emph{winning strategy} if every $f$-play of $G(\cM,M,v)$ is won
  by the player Automaton --- no matter what the moves of Pathfinder
  are.
\end{defi}

A (memoryless) (winning) strategy for player Pathfinder is defined
analogously. Note that we can consider a memoryless strategy $f$ as a
function $f:C_A\to C_P$.
Even though we require a strategy to be a total function, we often
define a concrete strategy only partially and assume that all other
(uninteresting) values are mapped to a fixed game position from~$C_P$.

In the following, let $\cM=(S,\delta,\iota,\rank)$ be an MSCA and
$(M,v)$ be a pointed MSC. Furthermore, let $(C_A,C_P,E_A,E_P)$ be the
arena of $G(\cM,M,v)$ and $(C_A,C_P^\#,E_A^\#,E_P^\#)$ be the arena of
$G(\cM^\#,M,v)$.
Firstly, let us state the fact that parity games enjoy memoryless
determinacy.
\begin{prop}[\cite{DBLP:conf/dagstuhl/Kusters01}]\label{prop:gamesAreDetermined}
  From any game position in $G(\cM,M,v)$, either Automaton or
  Pathfinder has a memoryless winning strategy.
\end{prop}
We now establish a connection between accepting runs of $\cM$ and winning
strategies of the player Automaton.
\begin{lem}\label{lemma:ifAccRunThenAutomatonWins}
  If $(M,v)$ is accepted by $\cM$, then Automaton has a winning
  strategy in the game $G(\cM,M,v)$.
\end{lem}
\begin{proof} % [of Lemma~\ref{lemma:ifAccRunThenAutomatonWins}]
  Let $\rho=(C,E,r,\mu,\nu)$ be an accepting run of $\cM$ on
  $(M,v)$. We construct a strategy~$f$ for Automaton which ensures
  that, for every $f$-play $\xi$ of $G(\cM,M,v)$, the label of $\xi$
  is also a label of a branch of $\rho$. Let $x\in C$ be a
  configuration with $E_\rho(x)\neq\emptyset$ and
  $\branch=x_1x_2\ldots x_n$ be the unique path from the root $r$ to
  $x$ in $\rho$. Consider the finite sequence
  \[\xi=(v_1,s_1)(v_1,\tau_1)(v_2,s_2)\ldots(v_n,s_n) \in (C_A
  C_P)^{n-1} C_A\]
  where $v_i=\nu(x_i)$, $s_i=\mu(x_i)$,
  $\tau_j=\mathsf{tr}_\rho(x_j)$ for all $i\in[n]$ and
  $j\in[n-1]$. The sequence~$\xi$ is a partial play of
  $G(\cM,M,v)$.  We define $f(\xi)=(v_n,\mathsf{tr}_\rho(x_n))$. The
  partial function $f$ becomes a total function and, therefore, a
  strategy for the player Automaton by mapping every value, for which
  we did not define $f$, to a fixed game position from $C_P$.
        
  We show that every $f$-play develops along a branch of $\rho$. Every
  play of $G(\cM,M,v)$ starts in the initial position
  $(\iota,v)=(\nu(r),\mu(r))$. For the induction step, let
  \[\xi=(v_1,s_1)(v_1,\tau_1)(v_2,s_2)\ldots(v_n,s_n)\in
  (C_A C_P)^{n-1} C_A\] 
  be a partial $f$-play and
  $\branch=x_1\ldots x_n\in C^n$ be the prefix of the branch of $\rho$
  such that $\mu(x_i)=s_i$, $\nu(x_i)=v_i$ for all $i\in[n]$ and
  $\mathsf{tr}_\rho(x_j)=\tau_j$ for all $j\in[n-1]$. If $\cM$ is
  stuck in $x_n$, then we have $E_\rho(x_n)=\emptyset$ and the player
  Automaton cannot proceed in $\xi$. Otherwise, we have
  $E_\rho(x_n)\neq\emptyset$ and $f(\xi)$ is defined. After Automaton's
  $f$-conform move we are at game position
  $f(\xi)=(v_n,\mathsf{tr}_\rho(x_n))\in C_P$. For every move
  $(D,s)\in\mathsf{tr}_\rho(x_n)$ of Pathfinder, there exists a
  configuration $x\in E_\rho(x_n)$ with $\mu(x)=s$ and
  $\eta_M(\nu(x_n),\nu(x))=D$.  Hence, every $f$-play $\xi$ develops
  along a branch $\branch$ of $\rho$.
  Since $\branch$ is accepting and $\branch$ and $\xi$ are labelled by
  the same sequence over $S$, $\xi$ is a play won by
  Automaton. Therefore, $f$ is a winning strategy of Automaton in the
  game $G(\cM,M,v)$.
\end{proof}
\begin{lem}\label{lemma:ifAutomatonWinsThenAccRun}
  If the player Automaton has a winning strategy in the game
  $G(\cM,M,v)$, then the pointed MSC $(M,v)$ is accepted by $\cM$.
\end{lem}
\begin{proof} % [of Lemma~\ref{lemma:ifAutomatonWinsThenAccRun}]
  Let us assume that there exists a winning strategy $f$ for Automaton
  in $G(\cM,M,v)$. By Prop.~\ref{prop:gamesAreDetermined}, we can
  assume that $f$ is memoryless. We inductively construct an accepting
  run $\rho$ of $\cM$ on $(M,v)$. Firstly, we set
  $\rho_1=(C_1,E_1,r,\mu_1,\nu_1)$ where $C_1=\{r\}$, $E_1=\emptyset$,
  $\mu_1(r)=\iota$, and $\nu_1(r)=v$. Now, let us assume that the
  $S$-labelled tree $\rho_i=(C_i,E_i,r,\mu_i,\nu_i)$ is already
  defined. Let $\{x_1,x_2,\ldots,x_n\}$ be the set of all leaves of
  $\rho_i$ in which $\cM$ is not stuck, i.e., for all $j\in[n]$, the
  local MSCA $\cM$ is not stuck in state $\mu_i(x_j)$ at position
  $\nu_i(x_j)$. For every $j\in[n]$, let $\tau_j$ be the transition
  such that $f(\nu_i(x_j),\mu_i(x_j))=(\nu_i(x_j),\tau_j)$. We set
  $\rho_{i+1}=(C_{i+1},E_{i+1},r,\mu_{i+1},\nu_{i+1})$ to the smallest
  (with respect to the size of the set of configurations $C_{i+1}$)
  $S$-labelled tree such that $\rho_{i+1}\restriction C_i=\rho_i$ and,
  for all $j\in[n]$, $\mathsf{tr}_{\rho_{i+1}}(x_j)=\tau_j$.
        
  Let $\rho=(C,E,r,\mu,\nu)=\bigcup_{i\geq1}\rho_i$. It can be easily
  checked that $\rho$ is a run of $\cM$ on $(M,v)$.  Now, let
  $\branch=x_1x_2x_3\ldots\in C^\infty$ be a branch of $\rho$. Consider
  the play
  \[\xi=\big(\nu(x_1),\mu(x_1)\big)\big(\nu(x_1),\mathsf{tr}_\rho(x_1)\big)\big(\nu(x_2),\mu(x_2)\big)\big(\nu(x_2),\mathsf{tr}_\rho(x_2)\big)\ldots\]
  of $G(\cM,M,v)$. It follows from the construction of $\rho$ that
  $\xi$ is an $f$-play.  Since $f$ is a winning strategy for player
  Automaton, $\xi$ is won by Automaton. Since $\xi$ and $\branch $ share
  the same label, the branch $\branch$ is accepting in $\cM$. Hence,
  $\rho$ is an accepting run of the local MSCA $\cM$.
\end{proof}
The next two lemmas state that a player has a winning strategy in the
current game if and only if there exists a winning strategy for its
opponent in the dual game.
\begin{lem}\label{lemma:ifAutomatonThenPathfinder}
  If Automaton has a winning strategy in $G(\cM,M,v)$, then
  Pathfinder has a winning strategy in the game $G(\cM^\#,M,v)$.
\end{lem}
\begin{proof}
  Let $f_A$ be a winning strategy of Automaton in the game
  $G(\cM,M,v)$. We show that there exists a strategy $f_P$ for
  Pathfinder such that, for every $f_P$-play $\xi^\#$ in
  $G(\cM^\#,M,v)$, there exists an $f_A$-play $\xi$ in $G(\cM,M,v)$
  such that $\xi\restriction C_A=\xi^\#\restriction C_A$. Note that
  the initial positions of the games $G(\cM,M,v)$ and $G(\cM^\#,M,v)$
  are the same. For the induction step, let $n\geq0$, $\xi^\#\in (C_A
  C_P^\#)^n \{(w,s)(w,\tau^\#)\}$ be a partial play of
  $G(\cM^\#,M,v)$, and $\xi\in(C_A C_P)^n\{(w,s)(w,\tau)\}$ be a
  partial $f_A$-play of $G(\cM,M,v)$ such that $\xi^\#\restriction
  C_A=\xi\restriction C_A$.  From Lemma~\ref{lemma:dualExpression} it
  follows that there exists a movement
  $(D,s')\in\tau\cap\tau^\#$. Pathfinder chooses $(D,s')$ as his next
  move resulting in a game position $(w',s')$ where $\eta^M(w,w')=D$,
  i.e., $f_P(\xi^\#)=(w',s')$. Clearly, the sequences $\xi(w',s')$ and
  $\xi^\#(w',s')$ are equal when restricting them to positions from
  $C_A$.

  Thus, for every $f_P$-play $\xi^\#$ in $G(\cM^\#,M,v)$, there exists an
  $f_A$-play $\xi$ in $G(\cM,M,v)$ such that $\mu(\xi^\#)=\mu(\xi)$. Since $f_A$
  is a winning strategy, $\xi$ is a play won by Automaton in $G(\cM,M,v)$. From
  Remark~\ref{remark:stateSequenceAcceptance} it follows that the play $\xi^\#$
  in $G(\cM^\#,M,v)$ is won by Pathfinder. Hence, we showed that Pathfinder has
  a winning strategy in the game $G(\cM^\#,M,v)$.
\end{proof}
\begin{lem}\label{lemma:ifPathfinderThenAutomaton}
  If Pathfinder has a winning strategy in $G(\cM^\#,M,v)$, then
  Automaton has a winning strategy in the game $G(\cM,M,v)$.
\end{lem}
\begin{proof}
  Let $f_P$ be a winning strategy of Pathfinder in the game
  $G(\cM^\#,M,v)$. We show that there exists a winning strategy $f_A$
  of Automaton ensuring that, for every $f_A$-play $\xi$ in the game
  $G(\cM,M,v)$, there exists an $f_P$-play $\xi^\#$ in $G(\cM^\#,M,v)$
  such that $\xi^\#\restriction C_A=\xi\restriction C_A$.  The initial
  positions of the games $G(\cM,M,v)$ and $G(\cM^\#,M,v)$ are the
  same. For the induction step, let $\xi\in (C_A C_P)^n\{(w,s)\}$ be a
  partial play of $G(\cM,M,v)$, and $\xi^\#\in (C_A
  C_P^\#)^n\{(w,s)\}$ be a partial $f_P$-play in $G(\cM^\#,M,v)$ such
  that $\xi\restriction C_A=\xi^\#\restriction C_A$.  We define
  \begin{multline*}
    X=\{(D,s')\in S\times\mathbb{M}\mid \text{there exists
      $\tau^\#\in\minmodel{\delta^\#(s,\lambda^M(w))}$ and $w'\in V^M$}\\
    \text{ such that }f_P(w,\tau^\#)=(w',s')\text{ and
    }\eta^M(w,w')=D\}
  \end{multline*}
  to be the set of the possible $f_P$-conform moves of Pathfinder
  after Automaton's next move in the play $\xi^\#$. We claim that there
  exists a transition $\tau\in\minmodel{\delta(s,\lambda^M(w))}$ with
  $\tau\subseteq X$.
        
  Towards a contradiction, suppose there is no such $\tau$. Then, for
  all $\tau'\in\minmodel{\delta(s,\lambda^M(w))}$, there
  exists a movement $(D_{\tau'},s_{\tau'})\in\tau'$ with
  $(D_{\tau'},s_{\tau'})\notin X$. If
  $Z=\{(D_{\tau'},s_{\tau'})\mid\tau'\in\minmodel{\delta(s,\lambda^M(w))}\}$,
  then $Z\cap Y\neq\emptyset$ for all
  $Y\in\model(\delta(s,\lambda^M(w)))$.  From
  Lemma~\ref{lemma:dualTransition} it follows that
  $Z\in\model(\delta^\#(s,\lambda^M(w)))$. Hence, there exists a
  transition $\tau''\in\minmodel{\delta^\#(s,\lambda^M(w))}$ with
  $\tau''\subseteq Z$. However, we have $\tau''\cap X=\emptyset$ which
  is a contradiction to our definition of $X$.
        
  Automaton chooses the above transition $\tau$ with $\tau\subseteq X$
  as his next move resulting in a game position $(w,\tau)$ in the game
  $G(\cM,M,v)$, i.e., $f_A(\xi)=(w,\tau)$. For every move $(D,s')$ of
  Pathfinder in $G(\cM,M,v)$, there exists a
  $\tau^\#\in\minmodel{\delta^\#(w,s)}$ with
  $f_P(\xi^\#(w,\tau^\#))=(w',s')$ and $\eta_M(w,w')=D$. This follows
  from the fact that $(D,s')\in X$. Clearly, the sequences
  $\xi(w,\tau)(w',s')$ and $\xi^\#(w,\tau^\#)(w',s')$ are equal when
  restricting them to positions from $C_A$.
        
  Let $\xi$ be an $f_A$-play in $G(\cM,M,v)$. There exists an
  $f_P$-play $\xi^\#$ in $G(\cM^\#,M,v)$ with $\xi\restriction
  C_A=\xi^\#\restriction C_A$. Since $\mu(\xi)=\mu(\xi^\#)$ and since
  $\xi^\#$ is a play won by Pathfinder in $G(\cM^\#,M,v)$, the play
  $\xi$ in $G(\cM,M,v)$ must be won by Automaton (by
  Remark~\ref{remark:stateSequenceAcceptance}). Thus, $f_A$ is a
  winning strategy for Automaton in $G(\cM,M,v)$.
\end{proof}
We are now able to prove our main theorem from this section.
\begin{proof}[Proof of Theorem~\ref{theorem:closedUnderComplement}]
  \label{proof:closedUnderComplement}
  By Lemma~\ref{lemma:ifAccRunThenAutomatonWins} and
  Lemma~\ref{lemma:ifAutomatonWinsThenAccRun}, the pointed MSC $(M,v)$
  is accepted by $\cM$ if and only if Automaton has a winning strategy
  in $G(\cM,M,v)$. By the lemmas \ref{lemma:ifAutomatonThenPathfinder}
  and \ref{lemma:ifPathfinderThenAutomaton}, the latter is the case if
  and only if Pathfinder has a winning strategy in the game
  $G(\cM^\#,M,v)$. From Prop.~\ref{prop:gamesAreDetermined} it follows
  that this is the case if and only if Automaton has no winning
  strategy in $G(\cM^\#,M,v)$ respectively $\cM^\#$ does not accept
  $(M,v)$ (again by the Lemmas \ref{lemma:ifAccRunThenAutomatonWins}
  and \ref{lemma:ifAutomatonWinsThenAccRun}).
\end{proof}

\section{Translation of Local CRPDL Formulas}
\label{sec:translationOfLocalFormulas}

In this section, we show that, for every local CRPDL formula $\alpha$,
one can compute a local MSCA $\cM_\alpha$ in polynomial time which
exactly accepts the set of models of $\alpha$. More formally:
\begin{thm}\label{theorem:localFormulasToAutomata}
  From a local formula $\alpha$, one can construct in time
  $\poly{|\alpha|}$ a local MSCA $\cM_\alpha$ such that, for all
  pointed MSCs $(M,v)$, we have $M,v\models\alpha$ if and only if
  $(M,v)\in\lang(\cM_\alpha)$. The size of $\cM_\alpha$ is linear in
  the size of $\alpha$.
\end{thm}
The rest of this section prepares the proof of the above
theorem. The actual proof can be found on
page~\pageref{proof:localFormulasToAutomata}.

\subsection{Construction}

\begin{figure}
  \begin{center}\small
    \pgfdeclarelayer{background}
    \pgfdeclarelayer{foreground}
    \pgfsetlayers{background,main,foreground}
    \begin{tikzpicture}[>=stealth',semithick,auto,shorten
      >=1pt,shorten <=1pt,initial text=,initial distance=0.7cm, every
      state/.style={fill=white,draw=black,text=black,
      minimum size=1cm,inner sep=1pt}]
      
      \begin{scope}
        \node[state,initial] (i) at (0,0) {$\iota\mid 1$};
        \node[state] (c) at (2.9,0) {$\cstate\mid 0$};
      \end{scope}
      
      \node[inner sep=0pt] (fork) at (0.75,0) {};
      
      \path[->] (i) edge (c);
        
      \node[anchor=south west] at ([xshift=0.4cm,yshift=0.1cm] i)
      {$\{\sigma\}$};
      
      \node[anchor=south east] at ([xshift=-0.4cm,yshift=0.1cm] c)
      {$\id$};
    \end{tikzpicture}
    \hspace{1.5cm}
    \begin{tikzpicture}[>=stealth',semithick,auto,shorten >=1pt,shorten
      <=1pt,initial text=,initial distance=0.7cm,
      every state/.style={fill=white,draw=black,text=black,
      minimum size=1cm,inner sep=1pt}]
      
      \begin{scope}
        \node[state,initial] (i) at (0,0) {$\iota\mid 1$};
        \node[state] (c) at (2.9,0) {$\cstate\mid 0$};
      \end{scope}
      
      \node[inner sep=0pt] (fork) at (0.75,0) {};
      
      \path[->] (i) edge (c);
        
      \node[anchor=south west] at ([xshift=0.45cm,yshift=0.1cm] i)
      {$\Sigma$};
      
      \node[anchor=south east] at ([xshift=-0.4cm,yshift=0.1cm] c)
      {$\proc$};
    \end{tikzpicture}
  \end{center}
  \caption{Illustrations of the local MSCAs $\cM_\sigma$ (left side) and
    $\cM_{\pth{\proc}\true}$ (right side).}
  \label{fig:simpleMSCAs}
\end{figure}
If $\alpha$ is a local formula,
then we distinguish the following cases:
\subsubsection*{Case $\alpha=\sigma$}

We define $\cM_\sigma=(\{\iota,\cstate\},\delta,\iota,\cstate,\rank)$
where $\rank(\iota)=1$, $\rank(\cstate)=0$, and
\[\delta(s,\sigma')=\begin{cases}(\id,\cstate)&\text{if
    $\sigma=\sigma'$ and
    $s=\iota$}\\
  \bot&\text{otherwise}
\end{cases}\]
for all $s\in\{\iota,\cstate\}$ and $\sigma'\in\Sigma$. The local MSCA
$\cM_\sigma$ is depicted on the left side of
Fig.~\ref{fig:simpleMSCAs}.

\subsubsection*{Case $\alpha=\neg\beta$}

We define $\cM_{\neg\beta}$ to be the dual automaton of $\cM_{\beta}$
(cf.~Definition~\ref{def:dual_msca}).

\subsubsection*{Case $\alpha=\pth{D}\true$ with $D\in\mathbb{M}$}

We define
$\cM_{\pth{D}\true}=(\{\iota,\cstate\},\delta,\iota,\cstate,\rank)$
where $\rank(\iota)=1$, $\rank(\cstate)=0$, and
\[\delta(s,\sigma)=\begin{cases}
  (D,\cstate)&\text{if $s=\iota$}\\
  \bot&\text{otherwise}
\end{cases}\] 
for all $s\in \{\iota,\cstate\}$ and $\sigma\in\Sigma$. On the right
side of Fig.~\ref{fig:simpleMSCAs}, there is an illustration of
$\cM_{\pth{\proc}\true}$.

%\vspace{0.5ex}
    
%Note that we will use the concatenation states of our MSCAs to
%concatenate MSCAs in order to obtain MSCAs which correspond to more
%complex formulas. For example, we can concatenate two copies of
%$\cM_{\pth{\proc}\true}$ to obtain a new automaton for the formula
%$\pth{\proc;\proc}\true$.
%
\begin{figure}
  \begin{center}\small
    \pgfdeclarelayer{background}
    \pgfdeclarelayer{foreground}
    \pgfsetlayers{background,main,foreground}
    \begin{tikzpicture}[>=stealth',semithick,auto,shorten
      >=1pt,shorten <=1pt, initial text=,initial distance=0.7cm,every
      state/.style={fill=white,draw=black,text=black,
      minimum size=1cm,inner sep=1pt}]

      \begin{scope}[every node/.style={cloud,cloud puffs=14,draw, 
        cloud ignores aspect,fill=white,minimum width=3.25cm,
        minimum height=2.25cm}]

        \node (one) at (0,0) {$\cM_{\pth{\pi_1}\true}$};
        \node (two) at (6,0) {$\cM_{\pth{\pi_2}\true}$};
      \end{scope}
      
      \begin{scope}
        \node[state,initial] (i1) at (one.west) {$\iota_1\mid 1$};
        \node[state] (c1) at (one.east) {$\cstate_1\mid 0$};
        
        \node[state] (i2) at (two.west) {$\iota_2\mid 1$};
        \node[state] (c2) at (two.east) {$\cstate_2\mid 0$};
      \end{scope}
      
      \path[->] (c1) edge (i2);
      \node[anchor=south west] at ([xshift=0.5cm,yshift=0.1cm]c1) {$\Sigma$};
      \node[anchor=south east] at ([xshift=-0.4cm,yshift=0.1cm]i2) {$\id$};
      
      % \begin{pgfonlayer}{background}
      %   \pgfsetcornersarced{\pgfpoint{5mm}{5mm}}
      %   \fill[draw=gray,fill=white!90!black] ([yshift=1.5cm]i1)
      %   rectangle
      %   ([yshift=-1.5cm]c2);
      % \end{pgfonlayer}
    \end{tikzpicture}
  \end{center}
  \caption{Illustration of $\cM_{\pth{\pi_1;\pi_2}\true}$.}
  \label{fig:concatenation}
\end{figure}
\subsubsection*{Case $\alpha=\pth{\pi_1;\pi_2}\true$}
    
If $\cM_{\pth{\pi_i}\true}=(S_i,\delta_i,\iota_i,\cstate_i,\rank_i)$
for $i\in[2]$ and $\delta_1(\cstate_1,\sigma)=\bot$ for all
$\sigma\in\Sigma$, then we define $\cM_{\pth{\pi_1;\pi_2}\true}$ to be
the local MSCA $(S,\delta,\iota_1,\cstate_2,\rank)$ where $S=S_1\uplus
S_2$, $\rank=\rank_1\cup \rank_2$, and
\[\delta(s,\sigma)=\begin{cases}
    (\id,\iota_2)&\text{if $s=\cstate_1$}\\
    \delta_1(s,\sigma)&\text{if $s\in S_1\setminus\{\cstate_1\}$}\\
    \delta_2(s,\sigma)&\text{if $s\in S_2$}
  \end{cases}
\]
for all $s\in S$ and $\sigma\in\Sigma$. Figure~\ref{fig:concatenation}
shows an illustration of $\cM_{\pth{\pi_1;\pi_2}\true}$.

% \vspace{0.5ex}

The automaton $\cM_{\pth{\pi_1;\pi_2}\true}$ is the concatenation of
the local MSCAs $\cM_{\pth{\pi_1}\true}$ and
$\cM_{\pth{\pi_2}\true}$. Intuitively, $\cM_{\pth{\pi_1;\pi_2}\true}$
starts a copy of $\cM_{\pth{\pi_1}\true}$ and, when this copy changes
into its concatenation state $\cstate_1$, the automaton
$\cM_{\pth{\pi_1;\pi_2}\true}$ proceeds with starting a copy of the
local MSCA $\cM_{\pth{\pi_2}\true}$. Note that
$\cM_{\pth{\pi_1;\pi_2}\true}$ is forced to start the copy of
$\cM_{\pth{\pi_2}\true}$ since runs of local MSCAs are maximal by
definition (see Definition~\ref{def:run}) and we have
$\{(\id,\iota_2)\}\in\minmodel{\delta(\cstate_1,\sigma)}$ for every
$\sigma\in\Sigma$.

\begin{figure}
  \begin{center}\small
    \pgfdeclarelayer{background}
    \pgfdeclarelayer{foreground}
    \pgfsetlayers{background,main,foreground}
    \begin{tikzpicture}[>=stealth',semithick,auto,shorten
      >=1pt,shorten <=1pt,initial text=,initial distance=0.7cm, every
      state/.style={fill=white,draw=black,text=black,
      minimum size=1cm,inner sep=1pt}]

      \begin{scope}[every node/.style={cloud,cloud puffs=14,draw,
        cloud ignores aspect,fill=white,minimum width=2.75cm,
        minimum height=2cm}]

        \node[anchor=west] (one) at (2,1.25) {$\cM_{\beta}$};
      \end{scope}
      
      \begin{scope}
        \node[state,initial] (i) at (0,0) {$\iota\mid 1$};
        \node[state] (c) at (6,0) {$\cstate\mid 0$};
        \node[state] (iprime) at (one.west) {$\iota'\mid 1$};
      \end{scope}
      
      \node[inner sep=0pt] (fork) at (0.75,0) {};
      
      \path[shorten >=-1pt, shorten <=0pt] (i) edge (fork);
      \path[shorten <=-1pt,->] (fork) edge (c)
        edge[bend right=35] (iprime);
        
      \node[anchor=south west] at ([xshift=0.5cm,yshift=0.1cm] i) {$\Sigma$};
      \node[anchor=south east] at ([xshift=-0.4cm,yshift=0.1cm] c) {$\id$};
      \node[anchor=north east] at ([xshift=-0.4cm,yshift=-0.1cm] iprime) {$\id$};
        
      % \begin{pgfonlayer}{background}
      %    \pgfsetcornersarced{\pgfpoint{5mm}{5mm}}
      %     \fill[draw=gray,fill=white!90!black] ([yshift=2.5cm] i) rectangle
      %     ([yshift=-0.75cm] c);
      % \end{pgfonlayer}
    \end{tikzpicture}
  \end{center}
  \caption{Illustration of the local MSCA $\cM_{\pth{\{\beta\}}\true}$.}
  \label{fig:test}
\end{figure}

\subsubsection*{Case $\alpha=\pth{\{\beta\}}\true$}

If $\cM_{\beta}=(S',\delta',\iota',\cstate',\rank')$, then we define
$\cM_{\pth{\{\beta\}}\true}=(S,\delta,\iota,\cstate,\rank)$ where
$S=S'\uplus\{\iota,\cstate\}$,
$\rank=\rank'\cup\{(\iota,1),(\cstate,0)\}$, and
\[\delta(s,\sigma)=\begin{cases}
    (\id,\iota')\land(\id,\cstate)&\text{if $s=\iota$}\\
    \bot&\text{if $s=\cstate$}\\
    \delta'(s,\sigma)&\text{if $s\in S'$}
  \end{cases}\]
for all $s\in S$ and $\sigma\in\Sigma$. The automaton
$\cM_{\pth{\{\beta\}}\true}$ is depicted in Figure~\ref{fig:test}.

%\vspace{0.5ex}

Intuitively, the local MSCA $\cM_{\pth{\{\beta\}}\true}$ starts $\cM_{\beta}$ to
test whether $M,v\models\beta$ holds and, at the same time, changes into its
concatenation state.

\begin{figure}
  \begin{center}\small
    \pgfdeclarelayer{background}
    \pgfdeclarelayer{foreground}
    \pgfsetlayers{background,main,foreground}
    \begin{tikzpicture}[>=stealth',semithick,auto,shorten
      >=1pt,shorten <=1pt, initial text=,initial distance=0.7cm,every
      state/.style={fill=white,draw=black,text=black,minimum size=1cm,inner sep=1pt}]

      \begin{scope}[every node/.style={cloud,cloud puffs=14,draw, 
        cloud ignores aspect,fill=white,minimum width=3.25cm,
        minimum height=1.85cm}]

        \node (one) at (0,0) {$\cM_{\pth{\pi_1}\true}$};
        \node (two) at (0,-2) {$\cM_{\pth{\pi_2}\true}$};
      \end{scope}
      
      \begin{scope}
        \node[state,initial] (i) at (-3.5,-1) {$\iota\mid 1$};
        \node[state] (c) at (3.5,-1) {$\cstate\mid 0$};

        \node[state] (i1) at (one.west) {$\iota_1\mid 1$};
        \node[state] (c1) at (one.east) {$\cstate_1\mid 0$};
        
        \node[state] (i2) at (two.west) {$\iota_2\mid 1$};
        \node[state] (c2) at (two.east) {$\cstate_2\mid 0$};
      \end{scope}
      
      \path[->] (i) edge (i1);
      \path[->] (i) edge (i2);

      \path[->] (c1) edge (c);
      \path[->] (c2) edge (c);

      \node[anchor=south west] at ([xshift=0.25cm,yshift=0.35cm]i) {$\Sigma$};
      \node[anchor=north west] at ([xshift=0.25cm,yshift=-0.35cm]i) {$\Sigma$}; 
      \node[anchor=north east] at ([xshift=-0.5cm,yshift=0.2cm]i1) {$\id$};
      \node[anchor=south east] at ([xshift=-0.5cm,yshift=-0.2cm]i2) {$\id$};

      \node[anchor=south east] at ([xshift=-0.25cm,yshift=0.4cm]c) {$\id$};
      \node[anchor=north east] at ([xshift=-0.25cm,yshift=-0.4cm]c) {$\id$};

      \node[anchor=north west] at ([xshift=0.5cm,yshift=0.2cm]c1) {$\Sigma$};
      \node[anchor=south west] at ([xshift=0.5cm,yshift=-0.2cm]c2) {$\Sigma$};
      
      % \begin{pgfonlayer}{background}
      %    \pgfsetcornersarced{\pgfpoint{5mm}{5mm}}
      %     \fill[draw=gray,fill=white!90!black] (-3.5,1) rectangle 
      %     (3.5,-3);
      % \end{pgfonlayer}
    \end{tikzpicture}
  \end{center}
  \caption{Illustration of $\cM_{\pth{\pi_1+\pi_2}\true}$.}
  \label{fig:plus}
\end{figure}

\subsubsection*{Case $\alpha=\pth{\pi_1+\pi_2}\true$}

If $\cM_{\pth{\pi_i}\true}=(S_i,\delta_i,\iota_i,\cstate_i,\rank_i)$
and $\delta_i(\cstate_i,\sigma)=\bot$ for all $i\in[2]$ and
$\sigma\in\Sigma$, then we define $\cM_{\pth{\pi_1+\pi_2}\true}$ to be
the local MSCA $(S,\delta,\iota,\cstate,\rank)$ where $S=S_1\uplus
S_2\uplus\{\iota,\cstate\}$,
$\rank=\rank_1\cup\rank_2\cup\{(\iota,1),(\cstate,0)\}$, and
\[\delta(s,\sigma)=\begin{cases}
    (\id,\iota_1)\lor(\id,\iota_2)&\text{if $s=\iota$}\\
    (\id,\cstate)&\text{if $s=\cstate_i$ and
    $i\in[2]$}\\
    \delta_i(s,\sigma)&\text{if $s\in S_i\setminus\{\cstate_i\}$ and
  $i\in[2]$}\\
\bot&\text{if $s=\cstate$}
  \end{cases}\]
for all $s\in S$ and $\sigma\in\Sigma$. The local MSCA
$\cM_{\pth{\pi_1+\pi_2}\true}$ is visualized in Fig.~\ref{fig:plus}.

\begin{figure}
  \begin{center}\small
    \pgfdeclarelayer{background}
    \pgfdeclarelayer{foreground}
    \pgfsetlayers{background,main,foreground}
    \begin{tikzpicture}[>=stealth',semithick,auto,shorten >=1pt,shorten
      <=1pt,initial text=,initial distance=0.7cm,
      every state/.style={fill=white,draw=black,text=black,
      minimum size=1cm,inner sep=1pt}]

      \begin{scope}[every node/.style={cloud,cloud puffs=14,draw,
        cloud ignores aspect,fill=white,minimum width=2.75cm,
        minimum height=2cm}]

        \node[anchor=center] (one) at (3,1.75) {$\cM_{\pth{\pi}\true}$};
      \end{scope}
      
      \begin{scope}
        \node[state,initial] (i) at (0,0) {$\iota\mid 1$};
        \node[state] (c) at (6,0) {$\cstate\mid 0$};
        \node[state] (iprime) at (one.west) {$\iota'\,|\,1$};
        \node[state] (cprime) at (one.east) {$\cstate'\mid 1$};
      \end{scope}
      
      \path[->] (i) edge (iprime);
      \path[->] (i) edge[bend right=15] (c);
      \draw[->] (cprime) .. controls +(2,-2) and +(1,0) .. (i);
        
      \node[anchor=south west] at ([xshift=0.5cm,yshift=0.1cm] i) {$\id$};
      \node[anchor=north west] at ([xshift=0.5cm,yshift=-0.3cm] i)
      {$\Sigma$};
      \node[anchor=south west] at ([xshift=0.1cm,yshift=0.5cm] i)
      {$\Sigma$};

      \node[anchor=north east] at ([xshift=-0.5cm,yshift=-0.3cm] c) {$\id$};
      \node[anchor=north east] at ([xshift=-0.5cm,yshift=-0.1cm] iprime)
      {$\id$};

      \node[anchor=north west] at ([xshift=0.6cm,yshift=-0.3cm] cprime)
      {$\Sigma$};
        
      % \begin{pgfonlayer}{background}
      %   \pgfsetcornersarced{\pgfpoint{5mm}{5mm}}
      %   \fill[draw=gray,fill=white!90!black] ([yshift=2.85cm] i)
      %   rectangle
      %   ([yshift=-0.85cm] c);
      % \end{pgfonlayer}
    \end{tikzpicture}
  \end{center}
  \caption{Illustration of the local MSCA $\cM_{\pth{\pi^\ast}\true}$.}
  \label{fig:star}
\end{figure}

\subsubsection*{Case $\alpha=\pth{\pi^\ast}\true$}

If $\cM_{\pth{\pi}\true}=(S',\delta',\iota',\cstate',\rank')$ and
$\delta'(\cstate',\sigma)=\bot$ for all $\sigma\in\Sigma$, then we set
$\cM_{\pth{\pi^\ast}\true}=(S,\delta,\iota,\cstate,\rank)$ where
$S=S'\uplus\{\iota,\cstate\}$, $\rank$ and $\rank'$ coincide on
$S'\setminus\{\cstate'\}$, $\rank'(s)=1$ if $s\in\{\iota,\cstate'\}$,
$\rank'(\cstate)=0$, and
\[\delta(s,\sigma)=\begin{cases}
    (\id,\iota)&\text{if $s=\cstate'$}\\
    (\id,\iota')\lor(\id,\cstate)&\text{if $s=\iota$}\\
    \bot&\text{if $s=\cstate$}\\
    \delta'(s,\sigma)&\text{if $s\in S'\setminus\{\cstate'\}$}
  \end{cases}\]
for all $s\in S$ and $\sigma\in\Sigma$. See Fig.~\ref{fig:star} for a
visualization of $\cM_{\pth{\pi^\ast}\true}$.
% (1) we don't want branches of height 1 and (2) an infinite path through
% iota may not be accepting; therefore, we don't have c(\iota)=0

% \vspace{0.5ex}

Intuitively, the local MSCA $\cM_{\pth{\pi^\ast}\true}$ executes a copy of
the automaton $\cM_{\pth{\pi}\true}$ and, every time this copy changes into
its concatenation state $\cstate'$, the local MSCA
$\cM_{\pth{\pi^\ast}\true}$ nondeterministically decides whether it restarts
this copy again or changes into the concatenation state $\cstate$.

\subsubsection*{Case $\alpha=\pth{\pi}^\omega$}
    
If $\cM_{\pth{\pi}\true}=(S,\delta',\iota,\cstate,\rank)$ and
$\delta'(\cstate,\sigma)=\bot$ for all $\sigma\in\Sigma$, then we set
$\cM_{\pth{\pi}^\omega}=(S,\delta,\iota,\cstate,\rank)$ where
\[\delta(s,\sigma)=\begin{cases}
    (\id,\iota)&\text{if $s=\cstate$}\\
    \delta'(s,\sigma)&\text{if $s\in S\setminus\{\cstate\}$}
  \end{cases}\]
for all $s\in S$ and $\sigma\in\Sigma$.

\subsection{Concatenation States}\label{subsec:concatenationStates}

In this section, we prove a technical proposition stating that, for
all path formulas $\alpha$, every accepting run of the local MSCA
$\cM_\alpha$ exhibits exactly one configuration labelled by the
concatenation state. It will be of use in
Sect.~\ref{subsec:correctness} to show the correctness of our
construction.

Firstly, we introduce the notion of main states. A state $s$ is called
a main state if the concatenation state can be reached from
$s$. The intuition of this type of states is the following: If $\pi$
is a path expression and $\rho$ is an accepting run of
$\cM_{\pth{\pi}\true}$, then $\rho$ exhibits one main branch $\branch$
by which $\cM$ ``processes'' the path expression $\pi$. The label of
$\branch$ solely consists of the not yet formally defined main
states. In all the other branches of $\rho$, i.e., in the branches
which fork from $\branch$, $\cM$ basically executes tests of the form
$\{\alpha\}$. All these branches are labelled by non-main states.
\begin{defi}
  Let $\cM=(S,\delta,\iota,\cstate,\rank)$ be a local MSCA and $s\in
  S$. We inductively define the set of \emph{main states}
  $\mainstates(\cM)$ of $\cM$:
  $\mainstates(\cM)$ is the least set such that, for all $s\in S$,
  we have $s\in\mainstates(\cM)$ if and only if
  \begin{enumerate}[(1)]
  \item $s=\cstate$ or
  \item there exist $s'\in\mainstates(\cM)$, $\sigma\in\Sigma$,
    $D\in\mathbb{M}$, and $\tau\in\minmodel{\delta(s,\sigma)}$ such
    that $(D,s')\in\tau$.
  \end{enumerate}
\end{defi}
By examining our construction, one can make the following two simple
observations.
\begin{rem}\label{remark:trivialFacts}
  If $\alpha$ is a path formula and
  $\cM_\alpha=(S,\delta,\iota,\cstate,\rank)$, the following
  conditions hold:
  \begin{enumerate}[(1)]
  \item we have $\delta(\cstate,\sigma)=\bot$ for every
    $\sigma\in\Sigma$
  \item\label{fact:mainStates} for all
    $s\in\mainstates(\cM_\alpha)$, we have
    \[\rank(s)=\begin{cases}
      0&\text{if $s=\cstate$}\\
      1&\text{otherwise}
    \end{cases}\]
  \end{enumerate}
\end{rem}\medskip

\noindent If $\alpha$ is a path formula of the form $\pth{\pi_1;\pi_2}\true$,
then in our construction of $\cM_\alpha$, we required
$\delta_1(\cstate_1,\sigma)=\bot$ for all $\sigma\in\Sigma$ where
$\delta_1$ is the transition relation and $\cstate_1$ is the
concatenation state of $\cM_{\pth{\pi_1}\true}$. It follows from the
above observation (1) that our construction can be applied to all
formulas of the form $\pth{\pi_1;\pi_2}\true$. Similarly, this holds
for our construction of~$\cM_\alpha$ in the cases
$\alpha=\pth{\pi_1+\pi_2}\true$, $\alpha=\pth{\pi^\ast}\true$, and
$\alpha=\pth{\pi}^\omega$.
\begin{lem}\label{lemma:mainStates}
  If $\alpha$ is a path formula and $\cM_\alpha=(S,\delta,\iota,\cstate,\rank)$,
  then the following two conditions hold:
  {\renewcommand*\theenumi{\alph{enumi}}
    \renewcommand*\labelenumi{(\theenumi)}
  \begin{enumerate}[\em(a)]
     \item $\iota\in\mainstates(\cM_\alpha)$
     \item for all $s\in\mainstates(\cM_\alpha)$,
       $\sigma\in\Sigma$, and $\tau\in\minmodel{\delta(s,\sigma)}$, we have
       \begin{equation}\label{eq:singleton}
         |\tau\cap(\mainstates(\cM_\alpha)\times\mathbb{M})|=1
       \end{equation}
   \end{enumerate}
   }
\end{lem}
Intuitively, the above lemma states that every run of $M_\alpha$
exhibits exactly one branch labelled solely by main states and that
all other configurations of this run which are not part of this path
are labelled by non-main states.
\begin{proof}
  By simple inspection, our claim follows for the cases $\alpha=\pth{D}\true$
  with $D\in\mathbb{M}$ and $\alpha=\pth{\{\beta}\}\true$.
  As our induction hypothesis, let us assume that the above lemma holds for
  $\cM_{\pth{\pi_i}\true}=(S_i,\delta_i,\iota_i,\cstate_i,\rank_i)$ where
  $i\in[2]$. 
  If $\alpha=\pth{\pi_1;\pi_2}\true$, then it can be easily checked
  that
  $\mainstates(\cM_\alpha)=\mainstates(\cM_{\pth{\pi_1}\true})\cup\mainstates(\cM_{\pth{\pi_2}\true})$. Hence,
  $\iota_1\in\mainstates(\cM_\alpha)$ and, therefore, property (a) is
  fulfilled. Now, let $\tau\in\minmodel{\delta(s,\sigma)}$ for some
  $s\in S$ and $\sigma\in\Sigma$. Then $\tau=\{(\id,\iota_2)\}$ (if
  $s=\cstate_1$), $\tau\in\minmodel{\delta_1(s,\sigma)}$, or
  $\tau\in\minmodel{\delta_2(s,\sigma)}$. Together with our induction
  hypothesis it follows that (\ref{eq:singleton}) holds for $\tau$.
  Now, let us consider the case $\alpha=\pth{\pi_1+\pi_2}\true$.  By
  easy inspection it follows that
  \[\mainstates(\cM_\alpha)=\{\iota,\cstate\} \cup
  \mainstates(\cM_{\pth{\pi_1}\true}) \cup
  \mainstates(\cM_{\pth{\pi_2}\true})\,.\] 
  Hence, property (a)
  follows. If $\tau\in\minmodel{\delta(s,\sigma)}$ for some $s\in S$
  and $\sigma\in\Sigma$, then $\tau=\{(\id,\iota_1)\}$,
  $\tau=\{(\id,\iota_2)\}$, $\tau=\{(\id,\cstate)\}$,
  $\tau\in\minmodel{\delta_1(s,\sigma)}$, or
  $\tau\in\minmodel{\delta_2(s,\sigma)}$. Property (b) follows from
  our induction hypothesis.

  Finally, we need to deal with the case
  $\alpha=\pth{\pi^\ast}\true$. For this, we assume that the above
  lemma holds for
  $\cM_{\pth{\pi}\true}=(S',\delta',\iota',\cstate',\rank')$. Again,
  it can be easily verified that
  $\mainstates(\cM_\alpha)=\{\iota,\cstate\} \cup
  \mainstates(\cM_{\pth{\pi}\true})$. Thus, property (a)
  holds. Now, let $\tau\in\minmodel{\delta(s,\sigma)}$ for some $s\in
  S$ and $\sigma\in\Sigma$. We have $\tau=\{(\id,\iota)\}$,
  $\tau=\{(\id,\iota')\}$, $\tau=\{(\id,\cstate)\}$, or
  $\tau\in\minmodel{\delta'(s,\sigma)}$. Property (b) follows from our
  induction hypothesis.
\end{proof}

\begin{prop}\label{prop:exactlyOneConcatenationConfig}
  Let $\alpha$ be a path formula and $\rho=(C,E,r,\mu,\nu)$ be an
  accepting run of
  $\cM_{\alpha}=(S,\delta,\iota,\cstate,\rank)$. There exists exactly
  one configuration from $C$ denoted by $\conc(\rho)$ with
  $\mu(\conc(\rho))=\cstate$.
\end{prop}

\begin{proof}
  It follows from Lemma~\ref{lemma:mainStates} that all configurations
  $x\in C$ with $\mu(x)\in\mainstates(\cM_\alpha)$ form a unique
  branch $\branch=x_1x_2x_3\ldots\in C^\infty$ of $\rho$. Since $\rho$
  is accepting, $\branch$ must be accepting. It follows from
  Remark~\ref{remark:trivialFacts} that $\mu(\branch)\in
  (\mainstates(\cM_\alpha)\setminus\{\cstate\})^\ast\{\cstate\}$.
  Therefore, every accepting run of $\cM_\alpha$ contains exactly one
  configuration labelled by~$\cstate$.
\end{proof}

\subsection{Correctness}\label{subsec:correctness}

Let $\alpha$ be a local formula. We show by induction over the
construction of $\alpha$ that $\lang(\cM_\alpha)=\lang(\alpha)$.  The
following claim is used as the induction hypothesis of our proof.
Recall that $\reach_M(v,\pi)$ is the set of all events which can be
reached from $v$ by a path described by $\pi$ in the MSC $M$.
\begin{myclaim}\label{claim:pathAutomata}
  Let $\alpha$ be a path formula. For all MSCs $M$, events $v,v'\in
  V^M$, we have $v'\in\reach_M(v,\alpha)$ if and only if there exists
  an accepting run $\rho$ of $\cM_\alpha$ on $(M,v)$ with
  $\nu(\conc(\rho))=v'$.
\end{myclaim}
The following four technical lemmas deal with the correctness of the
constructions of the local MSCAs $\cM_{\pth{\pi_1;\pi_2}\true} $ and
$\cM_{\pth{\pi^\ast}\true}$.
\begin{lem}\label{lemma:correctnessConcatenation1}
  Let $M$ be an MSC, $v_1,v'\in V^M$, and $\pi_1,\pi_2$ be path
  expressions. If Claim~\ref{claim:pathAutomata} holds for
  $\pth{\pi_1}\true$ and $\pth{\pi_2}\true$ and we have
  $v'\in\reach_M(v_1,\pi_1;\pi_2)$, then there exists an accepting run
  $\rho=(C,E,r,\mu,\nu)$ of $\cM_{\pth{\pi_1;\pi_2}\true}$ on
  $(M,v_1)$ with $\nu(\conc(\rho))=v'$.
\end{lem}
\begin{proof}
  Let $M_{\pth{\pi_1;\pi_2}\true}=(S,\delta,\iota,\cstate,\rank)$ and
  $M_{\pth{\pi_i}\true}=(S_i,\delta_i,\iota_i,\cstate_i,\rank_i)$ for
  all $i\in[2]$. If we have $v'\in\reach_M(v_1,\pi_1;\pi_2)$, then, by
  definition, there exists an event $v_2\in V^M$ such that
  $v_2\in\reach_M(v_1,\pi_1)$ and $v'\in\reach_M(v_2,\pi_2)$. It
  follows from our assumption that there exists an accepting run
  $\rho_1=(C_1,E_1,r_1,\mu_1,\nu_1)$ of the local MSCA
  $\cM_{\pth{\pi_1}\true}$ on $(M,v_1)$ with
  $\nu_1(\conc(\rho_1))=v_2$ and that there exists an accepting run
  $\rho_2=(C_2,E_2,r_2,\mu_2,\nu_2)$ of $\cM_{\pth{\pi_2}\true}$ on
  $(M,v_2)$ with $\nu_2(\conc(\rho_2))=v'$. Consider the $S$-labelled
  tree $\rho=(C_1\uplus C_2,E,r_1,\mu,\nu)$ where $E=E_1\cup
  E_2\cup\{(\conc(\rho_1),r_2)\}$, $\mu=\mu_1\cup\mu_2$, and
  $\nu=\nu_1\cup\nu_2$. It can be easily checked that $\rho$ is a run
  of the local MSCA $\cM_{\pth{\pi_1;\pi_2}\true}$ on $(M,v_1)$ with
  $\conc(\rho)=\conc(\rho_2)$. In Fig.~\ref{fig:run_conc}, the run
  $\rho$ is depicted where $\conc(\rho_i)$ is denoted by $x_i$ for
  $i\in[2]$.

  It remains to show that $\rho$ is accepting. Let $\branch$ be a
  branch of $\rho$. We distinguish two cases: If $\branch\in
  C_1^\infty$, then $\branch$ is also a branch from $\rho_1$. Since
  $\rho_1$ is accepting, the branch $\branch$ is accepting in
  $\cM_{\pth{\pi_1}\true}$. Since $\rank_1\subseteq\rank$ and
  $\mu_1\subseteq\mu$, $\branch$ is accepting in
  $\cM_{\pth{\pi_1;\pi_2}\true}$, too. Otherwise (i.e., if $\branch\in
  C_1^+\{r_2\}C_2^\infty$), there exists a suffix of $\branch$ which
  is an accepting branch in $\rho_2$. Because of this fact,
  $\mu_2\subseteq\mu$, and $\rank_2\subseteq\rank$, the branch
  $\branch$ is also accepting in
  $\cM_{\pth{\pi_1;\pi_2}\true}$. Hence, $\rho$ is an accepting run of
  the automaton $\cM_{\pth{\pi_1;\pi_2}\true}$ on $(M,v_1)$ with
  $\nu(\conc(\rho))=\nu(\conc(\rho_2))=v'$.
\end{proof}

\begin{figure}[tb]
  \begin{center}
    \begin{tikzpicture}[scale=0.8,font=\small,text height=1ex,text depth=.25ex]
      \draw (0,0) node[draw, shape=isosceles triangle, shape border rotate=180,
      anchor=apex, inner xsep=0.9 cm] (t1) {$\rho_1$}; 
      
      \draw (t1.apex) node[draw,fill=white,circle,inner sep=0.1ex]  {$r_1$};
      
      \draw (t1.apex) node [above,inner sep=1.45ex] {$\iota_1$};
      
      \draw ([yshift=0.3cm] t1.lower side) node[draw,circle,fill=white,inner
      sep=0.1ex] (x1) {$x_1$}; 
      
      \draw (x1) node [above,fill=white,inner sep=0.4ex,outer sep=1.45ex]
      {$\cstate_1$};
      
      \draw ([xshift=1cm] x1) node[draw, shape=isosceles triangle, shape border
      rotate=180, anchor=apex, inner xsep=0.7cm] (t2) {$\rho_2$};
      
      \draw (t2.apex) node[draw,fill=white,circle,inner sep=0.1ex]  {$r_2$};
      
      \draw (t2.apex) node [above,inner sep=1.45ex] {$\iota_2$};
      
      \draw ([yshift=-0.4cm] t2.lower side) node[draw,circle,fill=white,inner
      sep=0.1ex] (x2) {$x_2$};

      \draw (x2) node [above,fill=white,inner sep=0.4ex,outer sep=1.45ex]
      {$\cstate_2$};
    
      \begin{pgfonlayer}{background}
              \draw (x1) -- (t2.apex);
      \end{pgfonlayer}
    \end{tikzpicture}
  \end{center}
  \caption{The run $\rho$ of $\cM_{\pth{\pi_1;\pi_2}\true}$.}
  \label{fig:run_conc}
\end{figure}
        
\begin{lem}\label{lemma:correctnessConcatenation2}
  Let $M$ be an MSC, $v,v'\in V^M$, and $\pi_1,\pi_2$ be path
  expressions. If Claim~\ref{claim:pathAutomata} holds for
  $\pth{\pi_1}\true$ and $\pth{\pi_2}\true$ and there exists an
  accepting run $\rho=(C,E,r_1,\mu,\nu)$ of
  $\cM_{\pth{\pi_1;\pi_2}\true}$ on $(M,v)$ with
  $\nu(\conc(\rho))=v'$, then $v'\in\reach_M(v,\pi_1;\pi_2)$.
\end{lem}

\begin{proof}
  Let $M_{\pth{\pi_1;\pi_2}\true}=(S,\delta,\iota,\cstate,\rank)$ and
  $M_{\pth{\pi_i}\true}=(S_i,\delta_i,\iota_i,\cstate_i,\rank_i)$ for
  all $i\in[2]$.  Since $\mu(\conc(\rho))=\cstate$ and
  $\cstate=\cstate_2\in S_2$, there has to exist a configuration
  $r_2\in C$ with $\mu(r_2)=\iota_2$. Towards a contradiction, let us
  assume that there exists another configuration $r_3\in C$ with
  $\mu(r_3)=\iota_2$.  Because $\iota_2$ is a main state in
  $\cM_{\pth{\pi_1;\pi_2}\true}$ (see the proof of
  Lemma~\ref{lemma:mainStates}) and due to
  Lemma~\ref{lemma:mainStates}, $r_2$ and $r_3$ must occur in a branch
  of $\rho$. This is a contradiction to $\iota_2\notin
  \mathsf{src}_{\cM_{\pth{\pi_1;\pi_2}\true}}(\iota_2)$, i.e.,
  $\iota_2$ is not reachable from $\iota_2$ in
  $\cM_{\pth{\pi_1;\pi_2}\true}$. The latter fact follows by simple
  inspection of the transition relation of
  $\cM_{\pth{\pi_1;\pi_2}\true}$. Let $C_2 = \{y \in C \mid (r_2,y)
  \in E^\ast \}$ and $C_1 = C \setminus C_2$. It can be easily checked
  that the $S$-labelled tree $\rho_i=\rho\restriction C_i$ is a run of
  $\cM_{\pth{\pi_i}\true}$ for all $i\in[2]$. From the definition of
  the transition function $\delta$, it follows that there exists a
  configuration $x_1\in C_1$ with $\mu(x_1)=\cstate_1$,
  $E_\rho(x_1)=\{r_2\}$, $\nu(x_1)=\nu(r_2)$.
  Figure~\ref{fig:run_conc} shows a depiction of the run $\rho$
  consisting of $\rho_1$ and $\rho_2$ where $\conc(\rho)=x_2$.
        
  If $\branch$ is a branch of $\rho_2$ and $\branch'$ is the unique
  path in $\rho$ from $r_1$ to $x_1$, then $\branch'\branch$ is a
  branch of $\rho$. Since $\rho$ is accepting, $\mu_2\subseteq\mu$,
  and $\rank_2\subseteq\rank$, $\branch$ is accepting in
  $\cM_{\pth{\pi_2}\true}$. Hence, $\rho_2$ is an accepting run of
  $\cM_{\pth{\pi_2}\true}$ on $(M,\nu(x_1))$ with
  $\conc(\rho_2)=\conc(\rho)$.  Now, let $\branch$ be a branch of
  $\rho_1$. If $\branch\in(C\setminus\{x_1\})^\infty$, then $\branch$
  is a branch of $\rho$ with $\mu(\branch)\in S_1^\infty$. Since
  $\rho$ is accepting, $\mu_1\subseteq\mu$ and
  $\rank_1\subseteq\rank$, $\branch$ is accepting in
  $\cM_{\pth{\pi_1}\true}$. Otherwise (i.e., if $\branch\in
  C^\ast\{x_1\}$), $\branch$ ends in a configuration labelled by
  $\cstate_1$. By Remark~\ref{remark:trivialFacts},
  $\rank_1(\cstate_1)$ is even and, therefore, $\branch$ is accepting
  in $\rho_1$. Hence, $\rho_1$ is an accepting run of
  $\cM_{\pth{\pi_1}\true}$ on $(M,v)$ with $\conc(\rho_1)=x_1$. By our
  assumption, it follows that $\nu(x_1)\in\reach_M(v,\pi_1)$ and
  $v'\in\reach_M(\nu(x_1),\pi_2)$. Therefore, we have
  $v'\in\reach_M(v,\pi_1;\pi_2)$.
\end{proof}

\begin{lem}\label{lemma:correctnessStar1}
  Let $M$ be an MSC, $v,v'\in V^M$, and $\pi$ be a path expression. If
  Claim~\ref{claim:pathAutomata} holds for $\pth{\pi}\true$ and we
  have $v'\in\reach_M(v,\pi^\ast)$, then there exists an accepting run
  $\rho=(C,E,r,\mu,\nu)$ of $\cM_{\pth{\pi^\ast}\true}$ on $(M,v)$
  with $\nu(\conc(\rho))=v'$.
\end{lem}

\begin{proof}
  Let $\cM_{\pth{\pi^\ast}\true}=(S,\delta,\iota,\cstate,\rank)$ and
  $\cM_{\pth{\pi}\true}=(S',\delta',\iota',\cstate',\rank')$. Since
  $v'\in\reach_M(v,\pi^\ast)$, there exist an $n\geq0$ and events
  $v_1,v_2,\ldots,v_{n+1}\in V^M$ such that $v_1=v$, $v_{n+1}=v'$, and
  $v_{i+1}\in\reach_M(v_i,\pi)$ for all $i\in[n]$. If $n=0$, then the
  lemma follows by easy inspection of the construction of
  $\cM_{\pth{\pi^\ast}\true}$.  Now, let us assume that
  $n\geq1$. Since Claim~\ref{claim:pathAutomata} holds for
  $\pth{\pi}\true$, we can assume that there exist, for all $i\in[n]$,
  accepting runs $\rho_i=(C_i,E_i,r_i,\mu_i,\nu_i)$ of the automaton
  $\cM_{\pth{\pi}\true}$ on the pointed MSC $(M,v_i)$ with
  $\nu_i(\conc(\rho_i))=v_{i+1}$. Without loss of generality, we may
  assume that $C_i\cap C_j=\emptyset$ for all $i,j\in[n]$ (note that
  we can enforce $C_i \cap C_j = \emptyset$ by renaming the nodes of
  the $C_i$'s). Let $x_i=\conc(\rho_i)$ for all $i\in[n]$. The
  $S$-labelled tree $\rho=(C,E,y_1,\mu,\nu)$ where
  \begin{align*}
    C&=\textstyle\bigcup_{i\in[n]}C_i\uplus\{y_1,y_2,\ldots,y_{n+1},z\}\,\text{,}\\
    E&=\textstyle\bigcup_{i\in[n]}E_i\cup\{(y_i,r_i)\mid i\in[n]\}\cup\{(x_i,y_{i+1})\mid i\in[n]\}\cup\{(y_{n+1},z)\}\,\text{,}\\
    \mu&=\textstyle\bigcup_{i\in[n]}\mu_{i}\cup\{(y_i,\iota)\mid i\in[n+1]\}\cup\{(z,\cstate)\}\,\text{,}\\
    \nu&=\textstyle\bigcup_{i\in[n]}\nu_{i}\cup\{(y_i,v_i)\mid i\in[n+1]\}\cup\{(z,v')\}
  \end{align*}
  is a run of $\cM_{\pth{\pi^\ast}\true}$ on $(M,v)$ with
  $\conc(\rho)=v'$ --- this follows by an easy inspection of the
  construction of $\cM_{\pth{\pi^\ast}\true}$.
  Figure~\ref{fig:run_star} shows a depiction of $\rho$ for the case
  $n=3$.
        
  It remains to show that $\rho$ is accepting. Let $\branch$ be a branch
  of $\rho$. If $\branch\in(C\setminus\{z\})^\infty$, then there exist a
  suffix $\branch'$ of $\branch$ and an index $i\in[n]$ such that $\branch'$
  is a branch of $\rho_i$. Note that we have $\mu(\branch')\in
  (S'\setminus\{\cstate'\})^\infty$. Since $\rho_i$ is accepting,
  $\branch'$ is accepting in $\cM_{\pth{\pi}\true}$. Since
  $\mu_i\subseteq\mu$ and $\rank'\restriction
  (S'\setminus\{\cstate'\})\subseteq\rank$, it follows that $\branch$ is
  accepting in $\cM_{\pth{\pi^\ast}\true}$. Otherwise (i.e., if
  $\branch\in C^\ast\{z\}$), we have
  $\mu(\branch)=S^\ast\{\cstate\}$. Since $\rank(\cstate)$ is even,
  $\branch$ is accepting in $\cM_{\pth{\pi^\ast}\true}$. Hence, $\rho$
  is an accepting run of $\cM_{\pth{\pi^\ast}\true}$ on the
  pointed MSC $(M,v)$ with $\nu(\conc(\rho))=v'$.
\end{proof}

\begin{figure}[tb]
  \begin{center}
     \begin{tikzpicture}[scale=0.8,font=\small,text height=1ex,text depth=.25ex]
       \draw (0,0) node[draw,circle,fill=white,inner sep=0.1ex,minimum size=2.65ex] (y1) {$y_1$};
       \draw (y1) node [above,inner sep=1.75ex] {$\iota$};
  
       \foreach \y/\x/\s/\v in {2/1/0.6/-0.5, 3/2/0.8/0.25, 4/3/0.7/-0.4} {
         \draw ([xshift=0.8cm] y\x) node[draw, shape=isosceles triangle, 
           shape border rotate=180, anchor=apex, inner xsep=\s cm] (t\x) 
           {$\rho_\x$};
         \draw (t\x.apex) node[draw,fill=white,circle,inner sep=0.1ex]  {$r_\x$};
         \draw ([yshift=\v cm] t\x.lower side) node[draw,circle,fill=white,inner sep=0.1ex] (x\x) {$x_\x$};
         \draw ([xshift=0.8cm] x\x) node[draw,circle,fill=white,inner sep=0.1ex] (y\y) {$y_\y$};
         \draw (x\x) node [above,fill=white,inner sep=0.4ex,outer sep=1.45ex]
         {$\cstate'$};
         
         \draw (t\x.apex) node [above,inner sep=1.45ex] {$\iota'$};                              
         \draw (y\y) node [above,inner sep=1.45ex] {$\iota$};
       }
       
       \draw ([xshift=0.8cm] y4) node[draw,circle,fill=white,inner sep=0.35ex] (r4) {$z$};
       \draw (r4) node [above,inner sep=1.5ex] {$\cstate$};
       
       \begin{pgfonlayer}{background}
               \draw (y1) -- (t1.apex);
               \draw (x1) -- (y2) -- (t2.apex);
               \draw (x2) -- (y3) -- (t3.apex);
               \draw (x3) -- (y4) -- (r4);
       \end{pgfonlayer}        
     \end{tikzpicture}
  \end{center}
  \caption{The run $\rho$ of $\cM_{\pth{\pi^\ast}\true}$ consisting of three runs of $\cM_{\pth{\pi}\true}$ ($n=3$).}
  \label{fig:run_star}
\end{figure}

\begin{lem}\label{lemma:correctnessStar2}
  Let $M$ be an MSC, $v,v'\in V^M$, and $\pi$ be a path expression. If
  Claim~\ref{claim:pathAutomata} holds for $\pth{\pi}\true$ and there exists an
  accepting run $\rho=(C,E,y_1,\mu,\nu)$ of $\cM_{\pth{\pi^\ast}\true}$ on
  $(M,v)$ with $\nu(\conc(\rho))=v'$, then $v'\in\reach_M(v,\pi^\ast)$.
\end{lem}

\begin{proof}
  Let $\cM_{\pth{\pi^\ast}\true}=(S,\delta,\iota,\cstate,\rank)$ and
  $\cM_{\pth{\pi}\true}=(S',\delta',\iota',\cstate',\rank')$. Let $R$ be the set
  of all configurations from $C$ labelled by $\iota'$. If $R=\emptyset$, then
  $\rho$ consists of exactly one branch $\branch=y_1z$ with $\mu(\branch)=\iota
  \cstate$ and $\nu(y_1)=\nu(z)$. This can be easily verified by inspecting the
  transition function $\delta$.  From $\nu(y_1)=v$ and
  $\nu(z)=\nu(\conc(\rho))=v'$, it follows that $v=v'$. Hence,
  $v'\in\reach_M(v,\pi^\ast)$.

  Now, let us assume that $R\neq\emptyset$. It follows from
  $\iota'\in\mainstates(\cM_{\pth{\pi^\ast}\true})$ (see the last
  paragraph of the proof of Lemma~\ref{lemma:mainStates}) and
  Lemma~\ref{lemma:mainStates} that all configurations from $R$ occur
  in a unique finite branch $\branch=z_1z_2\ldots z_\ell$ of
  $\rho$. Without loss of generality, we can assume that
  $R=\{r_1,r_2,\ldots,r_n\}$ such that there exist
  $i_1<i_2<\ldots<i_n$ with $z_{i_k}=r_k$ for all $k\in[n]$. By
  examining the transition function $\delta$, one can see that:
  \begin{iteMize}{$\bullet$}
    \item For every $i\in[n]$, there exists a
      $y_i\in C$ with $E_\rho(y_i)=\{r_i\}$ and $\mu(y_i)=\iota$.
    \item There exists a configuration $y_{n+1}\in C$ with
      $E_\rho(y_{n+1})=\{\conc(\rho)\}$ and $\mu(y_{n+1})=\iota$. 
    \item For every $i\in[n]$, there exists a configuration
      $x_i$ with $E_\rho(x_i)=\{y_{i+1}\}$ and $\mu(x_i)=\cstate'$.
  \end{iteMize}
  Let $C_i = (\{ x \in C \mid (r_i, x) \in E^\ast\} \setminus \{ x \in
  C \mid (y_{i+1}, x) \in E^\ast\})$ for all $i\in[n]$. It can be
  easily verified that the $S$-labelled tree
  $\rho_i=(C_i,E_i,r_i,\mu_i,\nu_i)=\rho\restriction C_i$ is a run of
  $\cM_{\pth{\pi}\true}$ on $(M,\nu(r_i))$ with $\conc(\rho_i)=x_i$.
  Figure~\ref{fig:run_star} shows a depiction of the runs
  $\rho_1,\rho_2,\ldots,\rho_n$ forming the run $\rho$ for the case
  $n=3$.
        
  We now show that $\rho_i$ is an accepting run for every $i\in[n]$. Let
  $i\in[n]$ and $\branch$ be a branch of $\rho_i$. If $\branch\in
  (C_i\setminus\{x_i\})$, then there exists a path $\branch'$ from $y_1$ to $y_i$
  in $\rho$ such that $\branch'\branch$ is a branch of $\rho$. Since
  $\rho$ is accepting, $\branch'\branch$ is accepting in
  $\cM_{\pth{\pi^\ast}\true}$. Because of $\mu_i\subseteq\mu$, $\mu(\branch)\in
  (S'\setminus\{\cstate'\})^\infty$ and $\rank'\restriction
  (S'\setminus\{\cstate'\})\subseteq\rank$, $\branch$ is accepting
  in $\cM_{\pth{\pi}\true}$. If $\branch\in C^\ast\{x_i\}$, then $\mu(\branch)\in
  S'^\ast\{\cstate'\}$. By Remark~\ref{remark:trivialFacts},
  $\rank'(\cstate')$ is even and, therefore, $\branch$ is accepting in
  $\cM_{\pth{\pi}\true}$. Hence $\rho_i$ is an accepting run of
  $\cM_{\pth{\pi}\true}$ on $(M,\nu(r_i))$ with $\conc(\rho_i)=x_i$.
        
  Since Claim~\ref{claim:pathAutomata} holds for $\pth{\pi}\true$, we
  can assume that $\nu(x_i)\in\reach_M(\nu(r_i),\pi)$. By checking the
  definition of the transition function $\delta$, one can easily
  verify that $\nu(x_i)=\nu(y_{i+1})$ and $\nu(y_i)=\nu(r_i)$ holds
  for every $i\in[n]$. Hence, we have
  $\nu(y_{i+1})\in\reach_M(\nu(y_i),\pi)$ for every $i\in[n]$. From
  $\nu(y_{n+1})=\nu(\conc(\rho))$ it follows that
  $v'\in\reach_M(v,\pi^\ast)$.
\end{proof}
The following lemma dealing with the correctness of the construction
of $\cM_{\pth{\pi}^\omega}$ finishes the preparatory work needed in
order to proof Theorem~\ref{theorem:localFormulasToAutomata}.
\begin{lem}\label{lemma:correctnessOmega}
  Let $M$ be an MSC, $v\in V^M$, and $\pi$ be a path expression. If
  Claim~\ref{claim:pathAutomata} holds for $\pth{\pi}\true$, then
  \[M,v\models\pth{\pi}^\omega\iff (M,v)\in \lang(\cM_{\pth{\pi}^\omega})\]
\end{lem}

\begin{proof}
  Let us assume that $M,v\models\pth{\pi}^\omega$. There exist
  $v_1,v_2,v_3,\ldots\in V^M$ such that $v_1=v$ and
  $v_{i+1}\in\mathsf{reach}(v_i,\pi)$ for all $i\geq1$. Since
  Claim~\ref{claim:pathAutomata} holds for $\pth{\pi}\true$, there
  exists an accepting run $\rho_i=(C_i,E_i,r_i,\mu_i,\nu_i)$ of
  $\cM_{\pth{\pi}\true}$ on $(M,v_i)$ with
  $\nu(\conc(\rho_i))=v_{i+1}$ for every $i\geq1$. The $S$-labelled
  tree $\rho=(C,E,r_1,\mu,\nu)$ with $C=\biguplus_{i\geq1}C_i$,
  $E=\bigcup_{i\geq}E_i\cup\{(\conc(\rho_i),r_{i+1})\mid i\geq1\}$,
  $\mu=\bigcup_{i\geq1}\mu_i$, and $\nu=\bigcup_{i\geq1}\nu_i$ is a
  run of $\cM_{\pth{\pi}^\omega}$ on $(M,v)$. Let $\branch$ be a
  branch of $\rho$. If there exists an $i>1$ such that
  $\branch\in(C\setminus\{r_i\})^\infty$, then there exists a suffix
  $\branch'$ of $\branch$ such that $\branch'$ is an accepting branch
  of $\rho_j$ for some $j$ with $1\leq j<i$. Hence, $\branch$ is
  accepting in $\rho$. Otherwise (i.e., $\branch$ is a branch going
  through $r_i$ for every $i\geq1$), it follows from
  Remark~\ref{remark:trivialFacts} that we have $\min\{\rank(s)\mid
  s\in\mathsf{inf}(\branch)\}=\rank(\cstate)=0$. Hence, $\branch$ is
  accepting and, therefore, $\rho$ is accepting.
  The converse can be shown analogously.
\end{proof}

\bigskip

\noindent We are now able to prove our main theorem of this section.

\begin{proof}[Proof of
  Theorem~\ref{theorem:localFormulasToAutomata}]\label{proof:localFormulasToAutomata}
  By an easy analysis of our construction, one can see that, for all
  local formulas $\alpha$, the automaton $\cM_\alpha$ can be
  constructed in polynomial time and that its size is linear in the
  size of $\alpha$.

  Now, we inductively show that $\lang(\alpha)=\lang(\cM_\alpha)$ for
  every local formula $\alpha$.  Let us first consider the base
  cases. If $\alpha=\sigma$ with $\sigma\in\Sigma$, then it is easily
  checked that $\lang(\alpha)=\lang(\cM_\alpha)$. By simple
  inspection, it also follows that Claim~\ref{claim:pathAutomata}
  holds for $\alpha=\pth{D}\true$ with $D\in\mathbb{M}$.
  Regarding the induction step, we need to distinguish the following
  cases: If $\alpha=\neg\beta$, the claim follows from
  Theorem~\ref{theorem:closedUnderComplement}. By
  Lemma~\ref{lemma:correctnessOmega}, we have
  $\lang(\alpha)=\lang(\cM_\alpha)$ for $\alpha=\pth{\pi}^\omega$. 
  Claim~\ref{claim:pathAutomata} holds for
  $\alpha=\pth{\pi_1;\pi_2}\true$ and $\alpha=\pth{\pi^\ast}\true$
  because of the lemmas~\ref{lemma:correctnessConcatenation1},
  \ref{lemma:correctnessConcatenation2}, \ref{lemma:correctnessStar1},
  and \ref{lemma:correctnessStar2}. Analogously, it can be shown that
  Claim~\ref{claim:pathAutomata} is also true for the cases
  $\alpha=\pth{\{\beta\}}\true$ and
  $\alpha=\pth{\pi_1+\pi_2}\true$. Note that we have
  $M,v\models\alpha$ if and only if $\reach_M(v,\alpha)\neq\emptyset$
  for all pointed MSCs $(M,v)$. Hence,
  $\lang(\alpha)=\lang(\cM_\alpha)$ holds for the above path formulas.
\end{proof}

\section{Translation of Global CRPDL Formulas}
\label{sec:translationOfGlobalFormulas}

In this section, we demonstrate that, for every global
CRPDL formula~$\varphi$ of the form $\mathsf{E}\alpha$ or
$\mathsf{A}\alpha$, one can compute a global MSCA $\cG_\varphi$ in
polynomial time which exactly accepts the set of models of $\varphi$.
Let $\varphi$ be a global formula of the above form.

\begin{figure}
  \begin{center}\small
    \pgfdeclarelayer{background}
    \pgfdeclarelayer{foreground}
    \pgfsetlayers{background,main,foreground}
    \begin{tikzpicture}[>=stealth',semithick,auto,shorten >=1pt,shorten
      <=1pt,initial text=,initial distance=0.7cm,
      every state/.style={fill=white,draw=black,text=black,
      minimum size=1cm,inner sep=1pt}]

      \begin{scope}[every node/.style={cloud,cloud puffs=14,draw,
        cloud ignores aspect,fill=white,minimum width=2.75cm,
        minimum height=2cm}]

        \node[anchor=west] (one) at (3,0) {$\cM_{\alpha}$};
      \end{scope}
      
      \begin{scope}
        \node[state,initial] (i) at (0,0) {$\iota\mid 1$};
        \node[state] (iprime) at (one.west) {$\iota'\mid 1$};
        \node[state] (f) at (8,0) {$f\mid 0$};
      \end{scope}
      
      \path[->] (i) edge[loop above] node {$\Sigma,\proc$} (i);
      \path[->] (i) edge (iprime);
        
      \node[anchor=north west] at ([xshift=0.4cm,yshift=-0.1cm] i)
      {$\Sigma$};
      \node[anchor=north east] at ([xshift=-0.4cm,yshift=-0.1cm]
      iprime) {$\id$};
        
      % \begin{pgfonlayer}{background}
      %   \pgfsetcornersarced{\pgfpoint{5mm}{5mm}}
      %   \fill[draw=gray,fill=white!90!black] ([yshift=2.85cm] i)
      %   rectangle ([yshift=-0.75cm] one);
      % \end{pgfonlayer}
    \end{tikzpicture}
  \end{center}
  \caption{Illustration of the local MSCA $\cM_{\mathsf{E}\alpha}$.}
  \label{fig:mscaExists}
\end{figure}

\subsubsection*{Case $\varphi=\mathsf{E}\alpha$}

If $\cM_{\alpha}=(S',\delta',\iota',\cstate,\rank')$, then we set
$\cM_{\mathsf{E}\alpha}=(S,\delta,\iota,\cstate,\rank)$ where
$S=S'\uplus\{\iota,f\}$, $\rank(s)=\rank'(s)$ for all $s\in S'$,
$\rank(\iota)=1$, $\rank(f)=0$, and, for all $s\in S$ and
$\sigma\in\Sigma$,
\[\delta(s,\sigma)=\begin{cases}
  (\proc,\iota)\lor(\id,\iota')&\text{if $s=\iota$}\\
  \bot&\text{if $s=f$}\\
  \delta'(s,\sigma)&\text{otherwise}
\end{cases}\] 
Intuitively, the automaton $\cM_{\mathsf{E}\alpha}$ (depicted in
Fig.~\ref{fig:mscaExists}) moves forward on a process finitely many
times. At some event $v$, it nondeterministically decides to start the
automaton~$\cM_\alpha$ to check whether $(M,v)\models\alpha$ holds.

% \vspace{0.5ex}

Now, $\cG_{\mathsf{E}\alpha}=(\cM,I)$ is meant to work as follows: it
nondeterministically chooses a process on which it executes a copy of
$\cM_{\mathsf{E}\alpha}$ in state $\iota$. On all the other processes
it accepts immediately by starting $\cM_{\mathsf{E}\alpha}$ in the
sink state $f$ with rank $0$. More formally, we let
$\cG_{\mathsf{E}\alpha}=(\cM_{\mathsf{E}\alpha},I)$ where
\[
I=\{(s_1,s_2,\ldots,s_{|\mathbb{P}|})\mid\text{there exists $p\in\mathbb{P}$ such
  that $s_p=\iota$ and $s_q=f$ for all $p\neq q$}\}\,.
\]

\begin{figure}
  \begin{center}\small
    \pgfdeclarelayer{background}
    \pgfdeclarelayer{foreground}
    \pgfsetlayers{background,main,foreground}
    \begin{tikzpicture}[>=stealth',semithick,auto,shorten >=1pt,shorten
      <=1pt,initial text=,initial distance=0.7cm,
      every state/.style={fill=white,draw=black,text=black,
      minimum size=1cm,inner sep=1pt}]

      \begin{scope}[every node/.style={cloud,cloud puffs=14,draw,
        cloud ignores aspect,fill=white,minimum width=2.75cm,
        minimum height=1.9cm}]

        \node[anchor=west] (one) at (2,1.75) {$\cM_{\alpha}$};
      \end{scope}
      
      \begin{scope}
        \node[state,initial] (i) at (0,0) {$\iota_1\mid 1$};
        \node[state] (c) at (6,0) {$\iota_2\mid 0$};
        \node[state] (iprime) at (one.west) {$\iota'\mid 1$};
      \end{scope}
      
      \node[inner sep=0pt] (fork) at (0.75,0.2) {};
      
      \path[shorten >=-1pt, shorten <=0pt] (i) edge (fork);
      \path[shorten <=-1pt,->] (fork) edge[bend left=15] (c)
        edge[bend right=20] (iprime);
      \path[->] (c) edge[bend left=15] (i);
        
      \node[anchor=south west] at ([xshift=0.4cm,yshift=0.2cm] i)
      {$\Sigma$};
      \node[anchor=north west] at ([xshift=0.4cm,yshift=-0.3cm] i) {$\proc$};
      \node[anchor=south east] at ([xshift=-0.4cm,yshift=0.2cm] c) {$\id$};
      \node[anchor=north east] at ([xshift=-0.4cm,yshift=-0.1cm]
      iprime) {$\id$};
      \node[anchor=north east] at ([xshift=-0.4cm,yshift=-0.2cm] c) {$\Sigma$};
        
      % \begin{pgfonlayer}{background}
      %   \pgfsetcornersarced{\pgfpoint{5mm}{5mm}}
      %   \fill[draw=gray,fill=white!90!black] ([yshift=2.85cm] i)
      %   rectangle
      %   ([yshift=-0.75cm] c);
      % \end{pgfonlayer}
    \end{tikzpicture}
  \end{center}
  \caption{Illustration of the local MSCA $\cM_{\mathsf{A}\alpha}$.}
  \label{fig:mscaForAll}
\end{figure}

\subsubsection*{Case $\varphi=\mathsf{A}\alpha$}
    
If $\cM_\alpha=(S',\delta',\iota',\cstate,\rank')$, we set
$\cM_{\mathsf{A}\alpha}=(S,\delta,\iota_1,\cstate,\rank)$ where
$S=S'\uplus\{\iota_1,\iota_2\}$, $\rank(s)=\rank'(s)$ for all $s\in
S$, $\rank(\iota_1)=1$, $\rank(\iota_2)=0$, and
\[\delta(s,\sigma)=\begin{cases}
  (\id,\iota_2)\land(\id,\iota')&\text{if $s=\iota_1$}\\
  (\proc,\iota_1)&\text{if $s=\iota_2$}\\
  \delta(s,\sigma)&\text{otherwise}
\end{cases}\]
Informally speaking, the automaton $\cM_{\mathsf{A}\alpha}$ (depicted
in Fig.~\ref{fig:mscaForAll}) moves forward on a certain process $p$
and checks, for every event $v\in V_p^M$ of this process, if
$(M,v)\models\alpha$ holds. Note that, if $\cM_{\mathsf{A}\alpha}$ is
in state $\iota_2$ at an event $v$ such that there exists a successor
$v'$ of $v$ on the same process, then $\cM_{\mathsf{A}\alpha}$ is
forced to move to $v'$ and to change into the state $\iota_1$. That is
due to the fact that runs of local MSCAs are maximal by definition
(see Definition~\ref{def:run}) and because we have
$\{(\proc,\iota_1)\}\in\minmodel{\delta(\iota_2,\sigma)}$ for every
$\sigma\in\Sigma$.

% \vspace{0.5ex}

We define $\cG_{\mathsf{A}\alpha}=(\cM_{\mathsf{A}\alpha},I)$ where
$I=\{(\iota_1,\iota_1,\ldots,\iota_1)\}$. That means
$\cG_{\mathsf{A}\alpha}$ ensures $(M,v)\models\alpha$ for every $v\in
M$ by starting $\cM_{\mathsf{A}\alpha}$ in the state $\iota_1$ on
every process.

\bigskip

\noindent Using Theorem~\ref{theorem:localFormulasToAutomata} and by
simple inspection of the above construction, the following theorem can
be shown.
\begin{thm}\label{theorem:globalFormulasToAutomata}
  From a global formula $\varphi$ of the form
  $\varphi=\mathsf{E}\alpha$ or $\varphi=\mathsf{A}\alpha$, one can
  construct in time $\poly{|\varphi|}$ a global MSCA $\cG_\varphi$ such
  that, for all MSCs $M$, we have $M\models\varphi$ if and only if
  $M\in\lang(\cG_\varphi)$. The size of $\cG_\varphi$ is linear in
  the size of $\varphi$.
\end{thm}
If $\varphi$ is an arbitrary global formula, then we can also
construct an equivalent global MSCA $\cG_\varphi=(\cM,I)$. However,
this time the space needed for our construction is exponential in the
number of ``global'' conjunctions occurring in $\varphi$. In fact, the
size of $\cM$ is still linear in $\varphi$ but $|I|$ is exponential in
the number of conjunctive connectives occurring outside of subformulas
of the form $\mathsf{E}\alpha$ and $\mathsf{A}\alpha$, respectively.

When constructing a global MSCA from an arbitrary global formula, we
need to distinguish the following two additional cases:

\subsubsection*{Case $\varphi=\varphi_1\lor\varphi_2$}
    
Let $\cG_{\varphi_i}=(\cM_i,I_i)$ and $\cM_i=(S_i, \delta_i, \iota_i,
\rank_i)$ for all $i\in[2]$. Then we define
$\cG_{\varphi_1\lor\varphi_2}=(\cM,I)$ where $\cM=(S, \delta, \iota_1,
\rank)$, $S=S_1\uplus S_2$, $\delta=\delta_1\cup\delta_2$,
$\rank=\rank_1\cup\rank_2$, and $I=I_1\cup I_2$.

\subsubsection*{Case $\varphi=\varphi_1\land\varphi_2$}
  
Let $\cG_{\varphi_i}=(\cM_i,I_i)$ and $\cM_i=(S_i, \delta_i, \iota_i,
\cstate_i, \rank_i)$ for all $i\in[2]$.
We define $\cG_{\varphi_1\land\varphi_2}=(\cM,I)$ where
\[I=\big\{ \big( (s_1, s_1'), (s_2, s_2'), \ldots, (s_{|\mathbb{P}|},
s_{|\mathbb{P}|}') \big) \mid (s_1, s_2, \ldots, s_{|\mathbb{P}|})\in I_1, (s_1',
s_2', \ldots, s_{|\mathbb{P}|}') \in I_2 \big\} \, ,\]
$\cM=(S_1\uplus S_2\uplus S, \delta, \iota_1, \cstate_1, \rank)$,
$S=\{(s_1,s_2) \mid s_1\in S_1, s_2\in S_2\}$, $\rank = \rank_1 \cup
\rank_2 \cup \{(s,1) \mid s \in S\}$, and, for all $s\in S_1\cup
S_2\cup S$, $s_1 \in S_1$, $s_2 \in S_2$, and $\sigma\in\Sigma$:
\[
\delta(s,\sigma)=\begin{cases}
  (\id,s_1)\land(\id,s_2)&\text{if $s=(s_1, s_2) \in S$}\\
  \delta_1(s,\sigma)&\text{if $s\in S_1$}\\
  \delta_2(s,\sigma)&\text{if $s\in S_2$}
  \end{cases}
\]
Together with Theorem~\ref{theorem:globalFormulasToAutomata}, we obtain:
\begin{cor}
  From a global formula $\varphi$, one can construct in time
  $2^{\poly{|\varphi|}}$ a global MSCA~$\cG_\varphi$ such that, for
  all MSCs $M$, we have $M\models\varphi$ if and only if
  $M\in\lang(\cG_\varphi)$. The size of $\cG_\varphi$ is exponential
  in the size of $\varphi$.
\end{cor}

\section{The Satisfiability Problem}\label{sec:satisfiability}
We strive for an algorithm that decides, given a global formula
$\varphi$, whether $\lang(\varphi)\neq\emptyset$ holds. Unfortunately, the
satisfiability problem of CRPDL is undecidable. This follows from
results concerning Lamport diagrams which can be easily transferred to
MSCs \cite{DBLP:journals/cl/MeenakshiR04}. However, if one only
considers existentially $B$-bounded MSCs
\cite{DBLP:conf/forte/Peled00,DBLP:conf/fsttcs/MadhusudanM01,DBLP:journals/jcss/GenestMSZ06,DBLP:journals/iandc/GenestKM06},
then the problem becomes decidable. Intuitively, an MSC $M$ is
\emph{existentially $B$-bounded} if its events can be scheduled in
such a way that at every moment no communication channel contains more
than $B$ pending messages (see definition below). The rest of this
section prepares the proof of our main theorem which is stated
in the following. The proof itself can be found on
page~\pageref{page:proofMainTheorem}.
\begin{thm}\label{theorem:satisfiability}
  The following problem is PSPACE-complete:
  \begin{quote}
    \noindent Input: $B\in\N$ (given in unary) and a global CRPDL
    formula $\varphi$
  
    \noindent Question: Is there an existentially $B$-bounded MSC
    satisfying $\varphi$?
  \end{quote}
\end{thm}
\subsection{From MSCAs to Word Automata}\label{subsec:from_mscas_to_word_automata}
In order to be able to give uniform definitions of automata over MSCs
and words, respectively, we also consider words over an alphabet
$\Gamma$ as labelled relational structures. For this, we fix the set
$\mathbb{W}=\{\pr,\nx,\id\}$ of directions.
\begin{defi}
  Let $\Gamma$ be an arbitrary alphabet. A \emph{word-like structure
    over $\Gamma$} is a structure $W=(V^W,\nx^W,\lambda^W)$ where
  \begin{iteMize}{$\bullet$}
  \item $V^W$ is a set of \emph{positions}, 
  \item $\nx^W\subseteq (V^W \times V^W)$,
  \item $\lambda^W\colon V^W \to \Gamma$ is a labeling function,
  \item $\nx^W$ is the direct successor relation of a linear order
    $\preceq^W$ on $V^W$,
  \item $(V^W,\preceq^W)$ is finite or isomorphic to
    $(\mathbb{N},\leq)$
  \end{iteMize}
  The word-like structure $W$ induces a partial function $\eta_W\colon
  (V^W \times V^W) \to \mathbb{W}$. For all $v,v'\in V^W$, we define
  \[\eta_W(v,v')=\begin{cases}
    \nx&\text{if $(v,v')\in\nx^W$}\\
    \pr&\text{if $(v',v)\in\nx^W$}\\
    \id&\text{if $v=v'$}\\
    \text{undefined}&\text{otherwise}
    \end{cases}\]
\end{defi}
Every finite word $W = \gamma_1 \gamma_2 \ldots
\gamma_n\in\Gamma^\ast$ gives rise to a unique (up to isomorphism)
word-like structure $\overline{W}$ where $V^{\overline{W}} = [n]$,
$\nx^{\overline{W}}=\{(i,i+1)\mid 1\leq i<n\}$, and
$\lambda^{\overline{W}}(i) = \gamma_i$ for all $i\in[n]$. Analogously,
every infinite word $W \in \Gamma^\omega$ induces a word-like
structure $\overline{W}$. In the following, we identify $W$ and
$\overline{W}$ for every word $W\in\Gamma^\infty$.

We now formalize the notion of existentially $B$-bounded MSCs.
\begin{defi}
  If $M$ is an MSC and $W$ is a word, then $W$ is a
  \emph{linearization} of $M$ if $V^M = V^W$, $\lambda^M=\lambda^W$,
  and $\mathord(\msg^M\cup\bigcup_{p\in\mathbb{P}}\proc_p^M)^\ast \subseteq
  \mathord{\preceq^W}$.  The word $W$ is \emph{$B$-bounded} if we have
  \[|\{v'\mid v'\preceq^W v,\lambda^W(v')=p!q\}|-|\{v'\mid v'\preceq^W
  v,\lambda^W(v')=q?p\}|\,\leq\, B\,\text{,}\] 
  for every $v\in V^W$ and
  $(p,q)\in\mathsf{Ch}$. An MSC $M$ is \emph{existentially
    $B$-bounded} if there exists a $B$-bounded linearization of $M$,
  i.e., if it allows for an execution with $B$-bounded channels.
\end{defi}
\begin{exa}\label{ex:linearizations}
  Let $M$ be the MSC from Fig.~\ref{fig:MSC}. The word
  \[W = (1!2) \; (1!2) \; (1!2) \; (2?1) \; (2!1) \; (2?1) \; (2!1) \;
  (2?1) \; (2!1) \; (1?2) \; (1?2) \; (1?2) \; \in \; \Sigma^\ast\]
  is a $3$-bounded linearization of $M$. Note that parentheses have
  been introduced for readability. There is even a $1$-bounded
  linearization of $M$:
  \[W' = (1!2) \; (2?1) \; (1!2) \; (2!1) \; (2?1) \; (1!2) \; (1?2)
  \; (2!1) \; (2?1) \; (1?2) \; (2!1) \; (1?2) \; \in \; \Sigma^\ast\]
  Hence, $W'$ witnesses the fact that $M$ is existentially
  $1$-bounded.
\end{exa}
We define two-way alternating automata over words in the style of
local MSCAs.
\begin{defi}
  A \emph{two-way alternating parity automaton} (or \emph{2APA} for
  short) is a quadruple $\cP=(S,\delta,\iota,\rank)$ where
  \begin{iteMize}{$\bullet$}
  \item $S$ is a finite set of states,
  \item $\delta\colon(S\times\Sigma)\to\cB^+(\mathbb{W}\times S)$ is
    a transition function,
  \item $\iota\in S$ is an initial state, and
  \item $\rank\colon S\to\{0,1,\ldots,m-1\}$ is a ranking function
    with $m\in\mathbb{N}$.
  \end{iteMize}
  The \emph{size} of $\cP$ is $|S| + |\delta|$. If $|\tau| = 1$ for
  all $\tau \in \minmodel{\delta(s, \sigma)}$, $s \in S$, and $\sigma
  \in \Sigma$ (i.e., $\cP$ does not make use of universal branching),
  then $\cP$ is called a \emph{two-way parity automaton} (or
  \emph{2PA}). If $W$ is a word, then the definition of an
  $S$-labelled tree over $W$ is analogous to the definition of an
  $S$-labelled tree over a pointed MSC
  (cf. Definition~\ref{def:labelledTree}). Furthermore, an (accepting)
  run of a 2APA is defined in a similar way as it is defined for a
  local MSCA (cf.~Definitions~\ref{def:stuck}, \ref{def:run}, and
  \ref{def:accepting}).
  By $\lang(\cP)$, we denote the set of words~$W$ for which there
  exists an accepting run of $\cP$ on $(W,v)$ where $v$ is the minimal
  element from $V^W$ with respect to $\preceq^W$.
\end{defi}
Now, let us fix a channel bound $B\in\N$ and the alphabet
$\Gamma=\Sigma\times\{0,1,\ldots,B-1\}$.
\begin{defi}
  If $W$ is a $B$-bounded word over $\Sigma$, then we associate with
  $W$ the unique $B$-bounded word $W_B$ over $\Gamma$ where $V^W =
  V^{W_B}$, $\nx^W=\nx^{W_B}$, and, for every $v\in V^W$, we have
  $\lambda^{W_B}(v)=(\lambda^W(v),i)$ with $i=|\{v'\in V^W\mid
  v'\prec^W v,\lambda^W(v)=\lambda^W(v')\}|\bmod B$.
\end{defi}
That means that, in the second component of the labels in $W_B$,
we count events labelled by the same action modulo $B$. 
\begin{exa}
  Let $W$ and $W'$ be the words from
  Example~\ref{ex:linearizations}. For instance, $W_3$ is the word
  \[
    (1!2,0) \, (1!2,1) \, (1!2,2) \, (2?1,0) \, (2!1,0) \,
    (2?1,1) \, (2!1,1) \, (2?12) \, (2!1,2) \,
    (1?2,0) \, (1?2,1) \, (1?2,2)
  \]
  whereas $W_2'$ is given by:
  \[
    (1!2,0) \, (2?1,0) \, (1!2,1) \, (2!1,0) \, (2?1,1) \,
    (1!2,0) \, (1?2,0) \, (2!1,1) \, (2?1,0) \,
    (1?2,1) \, (2!1,0) \, (1?2,0)
  \]
\end{exa}
In $W_B$, we are able to quickly locate matching send and receive
events.  For example, if $v$ is a send event of $W_B$ labelled by
$(p!q,i)$, we just need to move to the smallest event $v'\in V^{W_B}$
(with respect to $\preceq^W$) with $v\preceq^{W_B} v'$ and
$\lambda^{W_B}(v')=(q?p,i)$.

%\subsection{Translation of Global MSCAs to 2APAs}
\medskip

Let $\cG=(\cM,I)$ be a global MSCA. We can construct a 2APA
$\cP_{\cG}=(S,\delta,\iota,\rank)$ that accepts exactly the set of
words $W_B$ where $W$ is a $B$-bounded linearization of an MSC
from $\lang(\cG)$. In order to construct $\cP_\cG$, there is one issue
which needs to be addressed.
%
% For the sake of clarity, we do not elaborate on the details of the
% construction of $\cP_{\cG}$ but rather present the underlying
% ideas. At the beginning, each run of $\cP_{\cG}$ splits into $|\mathbb{P}|$
% configurations which are the starting points of the simulations of
% the MSCAs of which the global MSCA $\cG$ is consisting.
%
Let $M$ be an MSC and $W$ be a $B$-bounded linearization of $M$. If
$v,v'\in V^M$ with $\eta_M(v,v')=\proc$, then a local MSCA is capable
of directly moving to $v'$. In general, this cannot be accomplished by
a 2APA running on $W_B$ since there may exist events $v''\in V^M$ with
$v\prec^{W_B} v''\prec^{W_B} v'$. To circumvent this limitation, the
idea is to introduce transitions which allow the 2APA to move forward
on $W_B$ and skip non-relevant events until it reaches the event
$v'$. Of course, we have to analogously deal with $\proci$, $\msg$,
and $\msgi$ transitions of local MSCAs.

More precisely, regarding the 2APA $\cP_\cG$, we use states of the
form $(s,p,\nx)$ to remember that we are searching for the next event
on process $p$ in the $\nx$-direction. In contrast, a state of the
form $(s,p!q,i,\pr)$ means that we are looking for the nearest send
event $p!q$ indexed by $i$ in the $\pr$-direction. The first component
is always used to remember the state from which we need to continue
the simulation of the local MSCA $\cM$ after finding the correct
event. If $\cM=(S',\delta',\iota',\rank')$, then the set of
states of $\cP_\cG$ is the following:
\begin{align*}
  S=&\{\iota,t\}\cup S' \cup \{(s,p,\pr),(s,p,\nx)\mid s\in S',p\in\mathbb{P}\}\\
  &\quad\cup\{(s,\sigma,i,\pr),(s,\sigma,i,\nx)\mid s\in
  S',\sigma\in\Sigma,0\leq i<B\}
\end{align*}
The intuition for the states from $I$ is as
follows: From the initial state $\iota$, the 2APA $\cP_\cG$
nondeterministically changes into a global initial state
$(\iota_1,\iota_2,\ldots,\iota_{|\mathbb{P}|})$ from $I$. That way, it simulates
$|\mathbb{P}|$ many copies of $\cM$ where the $p$-th copy
of $\cM$ is started in the state $\iota_p$ in the minimal event of
process $p$ (with respect to $\preceq_p^M$). More formally, for all
$\gamma\in\Gamma$, we define
\[\delta(\iota,\gamma)=\bigvee_{(\iota_1,\ldots,\iota_{|\mathbb{P}|})\in
  I}\big(\id,(\iota_1,1,\nx)\big)\land\big(\id,(\iota_2,2,\nx)\big)\land\ldots\land\big(\id,(\iota_{|\mathbb{P}|},|\mathbb{P}|,\nx)\big)\,.\]
Assume that the automaton $\cP_\cG$ is in a state of the form
$(s,p,D)$ resp. $(s,\sigma,i,D)$ at an event $v$. If
$\lambda^M(v)\notin\Sigma_p\times\{0,\ldots,B-1\}$
resp.~$\lambda^M(v)\neq(\sigma,i)$, i.e., if $v$ is not the event at
which the simulation of $\cM$ needs to be continued, then we stay in
the current state and move into direction $D$. Otherwise, we simulate
a transition $\tau\in\minmodel{\delta'(s,\lambda^M(v))}$ of the local
MSCA $\cM$ in the following manner: If $(\proc,s)\in\tau$, then we
change into the state $(s,p,\nx)$ and move along the
$\nx$-direction. If $(\proci,s)\in\tau$, then we act analogously in
the $\pr$-direction. Now, let us assume that $(\msg,s)\in\tau$. If
$\lambda^M(v)$ is of the form $(p!q,i)$, then we change into
$(s,q?p,i,\nx)$ and move along the $\nx$-direction. In contrast, if
$v$ is a receive event, then the local MSCA $\cM$ is unable to execute
the movement $(\msg,s)$. To simulate this behavior, we change into the
sink state $t$ and stay at $v$. From state $t$, the 2APA $\cP_\cG$ is
unable to accept. If $(\msgi,s)\in\tau$, then we proceed
similarly. Formally, for all $s\in S'$, $p\in\mathbb{P}$, $\sigma\in\Sigma$,
$i\in\{0,\ldots,B-1\}$, $D\in\mathbb{W}$, and $\gamma\in\Gamma$, we
have
\begin{align*}
  \delta\big((s,p,D),\gamma\big)&=\begin{cases}
    \big(D,(s,p,D)\big)&\text{if $\gamma\notin\Sigma_p\times\{0,\ldots,B-1\}$}\\
    (\id,s)&\text{if $\gamma\in\Sigma_p\times\{0,\ldots,B-1\}$}
  \end{cases}\\
  \delta\big((s,\sigma,i,D),\gamma\big)&=\begin{cases}
    \big(D,(s,\sigma,i,D)\big)&\text{if $\gamma\neq (\sigma,i)$}\\
    (\id,s)&\text{if $\gamma=(\sigma,i)$}
  \end{cases}\\
  \delta(s,\gamma)&=g(s,\gamma)\\
  \delta(t,\gamma)&=\bot
\end{align*}
where, for all $s\in S'$ and $\gamma=(p\theta q,i)\in\Gamma$,
$g(s,\gamma)$ is the positive Boolean expression which is obtained
from $\delta'(s,p\theta q)$ by applying the following substitutions:
for all $s'\in S'$, we exchange
\begin{iteMize}{$\bullet$}
\item $(\proc,s')$ by $\big(\nx,(s',p,\nx)\big)$,
\item $(\proci,s')$ by $\big(\pr,(s',p,\pr)\big)$,
\item $(\msg,s')$ by $(\id,t)$ if $\theta=?$,
\item  $(\msg,s')$ by $\big(\nx,(s',q?p,i,\nx)\big)$ if $\theta=!$,
\item $(\msgi,s')$ by $(\id,t)$ if $\theta=!$, and
\item  $(\msg,s')$ by $\big(\pr,(s',q!p,i,\pr)\big)$ if $\theta=?$.
\end{iteMize}
It remains to define the ranking function $\rank$ of $\cP_\cG$. For
all $s\in S'$, we define $\rank(s)=\rank'(s)$. If $s\in S\setminus
S'$, then we set $\rank(s)=m$ where $m$ is the smallest odd natural
number larger than $\max_{s\in S'}\rank'(s)$.
\begin{thm}\label{theorem:linearizationsAutomaton}
  Let $M$ be an MSC and $W$ some $B$-bounded linearization of $M$. We
  have $M\in\lang(\cG)$ if and only if $W_B\in\lang(\cP_{\cG})$. The size of
  $\cP_{\cG}$ is polynomial in $B$ and the size of~$\cG$.
\end{thm}
\begin{proof}[Proof sketch]
  If $\rho$ is a successful run of $\cP_{\cG}$ on an MSC $M$, then
  $\rho$ immediately splits into $|\mathbb{P}|$ many subtrees
  $\rho_q$. By easy inspection of the transition function of $\cP_\cG$
  it follows that there exists a global initial state
  $(\iota_1,\ldots,\iota_{|\mathbb{P}|})\in I$ such that, for every
  $q\in\mathbb{P}$, there is exactly one subtree
  $\rho_q=(C_q,E_q,r_q,\mu_q,\nu_q)$ with
  $\mu_q(r_q)=(\iota_q,q,\nx)$. Each of these subtrees~$\rho_q$ can be
  pruned in such a way that one obtains an accepting run $\rho_q'$ of
  $\cM$ starting in state~$\iota_q$ from the minimal event of $V^M_q$
  (with respect to $\preceq_q^M$). Thus, $M$ is accepted by
  $\cG$. Note that we obtain $\rho_q'$ from $\rho_q$ by essentially
  removing all configurations $x$ with $\mu_q(x)\notin S'$; of course,
  we need to update $E_q$ accordingly.

  The converse can be shown analogously. Basically, one only needs to
  pad and combine the accepting runs of the local MSCA $\cM$ on the
  different processes in order to obtain a successful run of
  $\cP_{\cG}$.
\end{proof}
\subsection{Checking the Emptiness of 2APAs}
In order to solve the emptiness problem for a 2APA $\cP$, we transform
$\cP$ into a B\"uchi automaton.
\begin{defi}
  Formally, a \emph{B\"uchi automaton} (or \emph{BA}) over the
  alphabet $\Sigma$ is a tuple $\cB = (S, \Delta, \iota, F)$ where $S$
  is a finite set of states, $\iota$ is the initial state, $F
  \subseteq S$ is the set of final states, and $\Delta \subseteq S
  \times \Sigma \times S$ is the transition relation. The \emph{size}
  of $\cB$ is $|S| + |\Delta|$. Let $W = \sigma_0 \sigma_1 \ldots \in
  \Sigma^\infty$ be a word of length $n \in \mathbb{N} \cup
  \{\infty\}$. The mapping $r \colon \mathbb{N} \to S$ is a \emph{run
    of~$\cB$ on $W$} if $r(0) = \iota$ and $(r(i), \sigma_i, r(i+1)\big)
  \in \Delta$ for all $i < n$. A word $W$ is \emph{accepted} by $\cB$
  if there exists a run $r$ such that $r(0)r(1)r(2)\ldots \in
  S^\infty$ is \emph{B\"uchi accepting}, i.e., if one of the following
  conditions is fulfilled:
  \begin{enumerate}[(1)]
    \item $n \in \mathbb{N}$ and $r(n) \in F$
    \item $\mathsf{inf}\big(r(0)r(1)r(2)\ldots\big) \cap F \neq \emptyset$
  \end{enumerate}
  By $\lang(\cB)$, we denote the set of words which are accepted by
  $\cB$.
\end{defi}
In contrast to common definitions of B\"uchi automata, item (1) allows
$\cB$ to accept finite words as well.
In the following, we also need to deal with two-way (alternating)
B\"uchi automata (2ABA and 2BA for short) which are defined analogously
to 2APA and 2PA but implement the B\"uchi acceptance condition instead
of the parity acceptance condition.

\begin{defi}
  More precisely, a \emph{two-way alternating B\"uchi automaton}
  (\emph{2ABA} for short) is a tuple $\cB = (S, \delta, \iota, F)$
  where $S$, $\delta$, and $\iota$ are defined as for 2APA's and $F
  \subseteq S$ is the set of final states. An (accepting) run of $\cB$
  is defined in a similar way as it is defined for a 2APA with the
  following modification: A sequence of states $(s_i)_{i\geq1} \in
  S^\infty$ is accepting if and only if it is B\"uchi accepting. A
  \emph{two-way B\"uchi automaton} (\emph{2BA}) is defined analogously
  to a 2PA.
\end{defi}
\begin{rem}\label{rem:parity_to_buechi}
Note that using the ideas from
\cite{DBLP:conf/fossacs/KingKV01}, a 2APA $\cP$ can be transformed
into a 2ABA $\cB$ in polynomial space such that the size of $\cB$
is polynomial in the size of $\cP$ and $\lang(\cP) = \lang(\cB)$.
\end{rem}
In \cite{DBLP:conf/lpar/DaxK08}, Dax and Klaedtke showed the following:
\begin{thm}[\cite{DBLP:conf/lpar/DaxK08}]\label{thm:alternation_elimination}
  From a 2APA $\cP$, one can construct a BA $\cB$ whose size is
  exponential in the size of $\cP$ such that $\lang(\cP) =
  \lang(\cB)$.
\end{thm}
%
% (the idea from DBLP:conf/fossacs/KingKV01 can be easily transferred
% to 2APA)
%
Note that Dax and Klaedtke actually stated that one can construct a BA
of size $2^{O((nk)^2)}$ where $n$ is the size of $\cP$ and $2k$ is the
maximal rank of a state from $\cP$. Since we can assume that the
maximal rank of a state of $\cP$ is linear in the number of states
of $\cP$, it follows that the size of $\cB$ is exponential in the size
of $\cP$. Furthermore, in \cite{DBLP:conf/lpar/DaxK08}, only infinite
words are considered. Nevertheless, it can be easily seen that the
result also applies to automata recognizing infinite and finite words
at the same time.

In the following, we recall parts of the proof of
Theorem~\ref{thm:alternation_elimination} and adapt it to our setting
in order to be able to prove Prop.~\ref{prop:emptiness}. Let $\cB =
(S, \delta, \iota, F)$ be an 2ABA over the alphabet $\Sigma$ and let
$\Gamma$ be an abbreviation of the function space $S \to 2^{\mathbb{W}
  \times S}$.  If $W$ is a word and $\rho = (C, E, r, \mu, \nu)$ is an
accepting run of $\cP$ on $W$, then the authors of
\cite{DBLP:conf/lpar/DaxK08} argue that we can assume without loss of
generality that all nodes $x$ and $y$ of $\rho$ with $\mu(x) = \mu(y)$
and $\nu(x) = \nu(y)$ exhibit isomorphic subtrees. Hence, $\rho$ can
be thought of as a directed acyclic graph (DAG) which can be
represented as a (possibly infinite) word of functions $f = f_1f_2
\ldots \in \Gamma^\infty$ where $f_j(q) = \mathsf{tr}_\rho(x)$
(cf.~Def.~\ref{def:run}), $\mu(x) = q$, and $\nu(x)$ is the $j$-th
position of $W$ with respect to $\nx^W$. From $\cB$, an intermediate
2BA $\cB' = ( S, \delta', \iota, S \setminus F)$ over the alphabet
$\Sigma \times \Gamma$ is constructed where, for all $s \in S$,
$\sigma \in \Sigma$, and $f \in \Gamma$, we have
\[\delta'\big(s, (\sigma, f)\big) =
\begin{cases}
  \bigvee_{(D, s') \in f(s)}(D, s') & \text{if }f(s) \in \minmodel{\delta(s, \sigma)}\\
  (\nx, s) & \text{otherwise.}
\end{cases}
\]
Note that the automaton $\cB'$ is of exponential size since the size
of the alphabet $\Gamma$ is exponential in the size of $\cB$. However,
the set of states of $\cB'$ equals the set of states of $\cB$. It is
shown that $\cB'$ rejects exactly those words $(\sigma_0,
f_0)(\sigma_1, f_1) \ldots \in (\Sigma \times \Gamma)^\infty$ where
the function word $(f_i)_{i\geq0}$ represents an accepting run of
$\cB$ on $(\sigma_i)_{i\geq0}$. In the course of the proof of
Theorem~\ref{thm:alternation_elimination}, using \cite[Theorem
4.3]{DBLP:conf/popl/Vardi88}, a B\"uchi automaton $\cB''$ whose size
is exponential in $\cB$ with $\lang(\cB'') = (\Sigma \times
\Gamma)^\infty \setminus \lang(\cB')$ is constructed. It is shown that
the projection of $\lang(\cB'')$ to the alphabet $\Sigma$ equals
$L(\cB)$.

We also recall the essential parts of the proof of Theorem 4.3 of
\cite{DBLP:conf/popl/Vardi88}, apply a minor correction and adapt it
to our setting.
%
% In \cite{DBLP:journals/ipl/Vardi89} a proof of a variant of this
% theorem for finite words can be found which is more detailed than
% the one in \cite{DBLP:conf/popl/Vardi88}.
%
Let $\cB = (S, \delta, \iota, F)$ be a 2BA. By $\mathsf{bwl}(\cB)$, we
denote the set $S^2 \times \{0, 1\}$. Intuitively, a triple $(s, t, b)
\in \mathsf{bwl}(\cB)$ expresses that at the current position there is
a backward loop starting in state $s$ and ending in state $t$. We have
$b = 1$ if and only if this loop visits a final state. A word
$(\sigma_1 \sigma_2
\ldots, m_0m_1\ldots, n_0n_1\ldots) \in (\Sigma \times
\mathsf{bwl}(\cB) \times 2^S)^\infty$ of length $h \in \mathbb{N} \cup
\{\infty\}$ is \emph{$\cB$-legal} if and only if there exists a
sequence $\ell_0\ell_1\ldots \in \mathsf{bwl}(\cB)^\infty$ of length
$h$ such that the following conditions are fulfilled:
\begin{iteMize}{$\bullet$}
\item $(s, t, 0) \in \ell_i$ if and only if either $\{(\id, t)\} \in
  \minmodel{\delta(s, \sigma_i)}$ or $i \geq 1$ and there are states $s',
  t' \in S$ and $b \in \{0, 1\}$ such that $(s', t', b) \in m_{i-1}$,
  $\{(s', \pr)\} \in \minmodel{\delta(s, \sigma_i)}$, and $\{(t, \nx)\} \in
  \minmodel{\delta(t', \sigma_{i-1})}$
\item $(s, t, 1) \in \ell_i$ if and only if either $\{(\id, t)\} \in
  \minmodel{\delta(s, \sigma_i)}$ and $t \in F$ or $i \geq 1$ and
  there are states $s', t' \in S$ and $b \in \{0, 1\}$ such that $(s',
  t', b) \in m_{i-1}$, $\{(s', \pr)\} \in \minmodel{\delta(s,
    \sigma_i)}$, $\{(t, \nx)\} \in \minmodel{\delta(t',
    \sigma_{i-1})}$, and in addition either $b = 1$ or $\{s', t', t\}
  \cap F \neq \emptyset$
\item $(s, t, 0) \in m_i$ if and only if there are $s_0, s_1,
  \ldots, s_k \in S$ and $b_0,b_1, \ldots, b_{k-1} \in \{0,1\}$ with
  $k > 0$ such that $s_0 = s$, $s_k = t$, and $(s_j, s_{j+1}, b_j) \in
  \ell_i$ for all $0 \geq j > k$
\item $(s, t, 1) \in m_i$ if and only if there are $s_0, s_1, \ldots,
  s_k \in S$ and $b_0, b_1, \ldots, b_{k-1} \in \{0,1\}$ with $k > 0$
  such that $s_0 = s$, $s_k = t$, $(s_j, s_{j+1}, b_j) \in \ell_i$ for
  all $0 \geq j > k$, and 
  % either
  $\{b_0, b_1, \ldots, b_{k_1}\} \cap \{1\} \neq \emptyset$
  % or $\{s_1, s_2, \ldots, s_k\} \cap F \neq \emptyset$
\item $s \in n_i$ if and only if there exists a state $s' \in S$ and $b
  \in \{0, 1\}$ such that $(s, s', b) \in m_i$ and one of the
  following conditions holds:
  \begin{iteMize}{$-$}
  \item $(s', s', 1) \in m_i$
  \item $s' \in F$ and $\cB$ cannot make a transition at position $i$
    in state $s'$
  \item $i \geq 1$ and there exists a state $s'' \in S$ such that
    $\{(s'', \pr)\} \in \minmodel{\delta(s', \sigma_i)}$ and $s'' \in
    n_{i-1}$
  \end{iteMize}
\end{iteMize}
Note that the $\ell_i$'s are only used to simplify the definition of
the $m_i$'s. The introduction of the $n_i$'s is a minor correction of
the proof of Theorem 4.3. Intuitively, we have $s \in n_i$ if there
exists a position $j \leq i$ and a state $s' \in S$ such that there
exists a backward run starting in $s$ allowing $\cB$ to visit the
$j$-th position of the input word in state $s'$ such that the
following holds: either $s' \in F$ and $\cB$ cannot make a transition
at position $j$ in state $s'$ or, at position $j$ in state $s'$, the
automaton $\cB$ can enter infinitely often a loop containing a final
state. Note that without the information contained in the $n_i$'s, we
would not capture accepting runs of $\cB$ which do not visit all
positions of the input word but, at some position $i$, go backward and
then accept without returning to $i$ again.

From the 2BA $\cB = (S, \delta, \iota, F)$, we can construct a BA
$\cB_1$ recognizing the set of all $\cB$-legal words. Let $\cB_1 =
(S_1, \Delta_1, \iota_1, F_1)$ be the BA where $S_1 = 2^{S^2} \times
\mathsf{bwl}(\cB) \times 2^S$, $\iota_1 = (\emptyset, \emptyset,
\emptyset)$, $F_1 = S_1$ and, for all $(p', \overline{m}',
\overline{n}'), (p, \overline{m}, \overline{n}) \in S_1$ and $(\sigma,
m, n) \in \Sigma\times \mathsf{bwl}(\cB) \times S$, we have $\big((p',
\overline{m}', \overline{n}'), (\sigma, m, n), (p, \overline{m},
\overline{n})\big) \in \Delta_1$ if and only if there exists $\ell
\subseteq \mathsf{bwl}(\cB)$ such that the following conditions hold:
\begin{iteMize}{$\bullet$}
\item $\overline{m} = m$, $\overline{n} = n$,
\item $(s,t) \in p$ if and only if $\{(\nx, t)\} \in
  \minmodel{\delta(s, \sigma)}$
\item $(s, t, 0) \in \ell$ if and only if $\{(\id, t)\} \in
  \minmodel{\delta(s, \sigma)}$ or there are states $s',t' \in S$ and $b
  \in \{0,1\}$ such that $(s', t', b) \in \overline{m}'$, $\{(\pr,
  s')\} \in \minmodel{\delta(s, \sigma)}$, and $(t', t) \in p'$
\item $(s, t, 1) \in \ell$ if and only if $\{(\id, t)\} \in
  \minmodel{\delta(s, \sigma)}$ and $t \in F$ or there are states
  $s',t' \in S$ and $b \in \{0,1\}$ such that $(s', t', b) \in
  \overline{m}'$, $\{(\pr, s')\} \in \minmodel{\delta(s, \sigma)}$,
  $(t', t) \in p'$, and in addition either $b = 1$ or $\{s', t', t\}
  \cap F \neq \emptyset$
\item $(s, t, 0) \in m$ if and only if there are $s_0, s_1, \ldots,
  s_k \in S$ and $b_0, b_1, \ldots, b_{k-1} \in \{0,1\}$ with $k > 0$
  such that $s_0 = s$, $s_k = t$, and $(s_j, s_{j+1}, b_j) \in \ell$
  for all $0 \geq j > k$
\item $(s, t, 1) \in m$ if and only if there are $s_0, s_1, \ldots,
  s_k \in S$ and $b_0,b_1, \ldots, b_{k-1} \in \{0,1\}$ with $k > 0$
  such that $s_0 = s$, $s_k = t$, $(s_j, s_{j+1}, b_j) \in \ell$ for
  all $0 \geq j > k$, and
  % either
  $\{b_0, b_1, \ldots, b_{k-1}\} \cap \{1\} \neq \emptyset$
  % or $\{s_1, s_2, \ldots, s_k\} \cap F \neq \emptyset$
\item $s \in n$ if and only if there exists a state $s' \in S$ and $b
  \in \{0, 1\}$ such that $(s, s', b) \in m$ and one of the following holds:
  \begin{iteMize}{$-$}
  \item $(s', s', 1) \in m$
  \item $s' \in F$ and $\cB$ cannot make a transition at the current
    position in state $s'$
  \item there exists a state $s'' \in S$ such that $\{(s'', \pr)\} \in
    \minmodel{\delta(s', \sigma)}$ and $s'' \in \overline{n}'$
  \end{iteMize}
\end{iteMize}
It remains to specify a BA $\cB_2$ such that a $\cB$-legal word
$(\sigma_0 \sigma_1 \ldots, m_0m_1\ldots, n_0 n_1 \ldots) \in (\Sigma
\times \mathsf{bwl}(\cB) \times 2^S)^\infty$ is accepted by $\cB_2$ if
and only if $(\sigma_i)_{i \geq 0}$ is accepted by the 2BA $\cB$. Let
$\cB_2 = (S_2, \Delta_2, \iota_2, F_2)$ where 
\begin{iteMize}{$\bullet$}
\item $S_2 = (S \cup \{\bot\})
\times \{0, 1\}$, 
\item $\iota_2 = (\iota, b)$ with $b = 1$ if and only if $\iota \in
  F$,
\item $F_2 = (S \cup \{\bot\}) \times \{1\}$, and, 
\item 
  % for all $(s, b), (s', b') \in S_2$ and $(\sigma, m, n) \in
  % (\Sigma \times \mathsf{bwl}(\cB) \times 2^S)$,
  we have $\big((s, b), (\sigma, m, n), (s', b')\big) \in \Delta_2$ if
  and only if one of the following conditions is fulfilled:
  \begin{iteMize}{$-$}
  \item there exist $s'' \in S$ and $b'' \in \{0, 1\}$ such that $(s,
    s'', b'') \in m$, $\{(s', \nx)\} \in \minmodel{\delta(s'',
      \sigma)}$, and ($b' = 1$ if and only if $b'' = 1$ or $s' \in F$)
  \item $s \in n$ and $s' = \bot$
  \item $s' = t' = \bot$
  \end{iteMize}
\end{iteMize}
Intuitively, the states of $\cB_2$ come with a flag. The flag is set
to $1$ if and only if the simulated automaton $\cB$ just visited a
final state. It can be shown that the projection of $\lang(\cB_1) \cap
\lang(\cB_2)$ to the alphabet $\Sigma$ equals the language of $\cB$.
\begin{prop}\label{prop:emptiness}
  If $\cP$ is a 2APA, then one can check the emptiness of $\lang(\cP)$ in
  polynomial space.
\end{prop}

\begin{proof}[Proof sketch]
  By Remark~\ref{rem:parity_to_buechi}, we can transform $\cP$ in
  polynomial space into a 2ABA recognizing the same language. By
  Theorem~\ref{thm:alternation_elimination}, we can construct a BA
  $\cB' = (S, \Delta, \iota, F)$ whose size is exponential in the size
  of $\cP$ such that $\lang(\cB') = \lang(\cP)$. Clearly, remembering
  a state of $\cB'$ requires only polynomial space. By inspecting the
  construction of $\cB'$, one can see that $\cB'$ can be obtained in
  space polynomial in the size of $\cP$. This means in particular:
  given two states $s, t \in S$ and $\sigma \in \Sigma$, one can check
  in polynomial space whether $(s, \sigma, t) \in \Delta$ holds. Since
  $\lang(\cB')$ is non-empty if there exists a final state $s \in F$
  which is reachable from $\iota$ (recall that our B\"uchi automata
  also accept finite words), the emptiness problem of $\cP$ can be
  solved in polynomial space.
\end{proof}
\subsection{The Decision Procedure}
We are now able to prove our main theorem:
\begin{proof}[Proof of
  Theorem~\ref{theorem:satisfiability}]\label{page:proofMainTheorem}
  The global formula $\varphi$ is a positive Boolean combination of
  global formulas $\varphi_1,\ldots,\varphi_n$ where, for every
  $i\in[n]$, $\varphi_i$ is of the form $\mathsf{A}\alpha_i$ or
  $\mathsf{E}\alpha_i$ for some local formula $\alpha_i$. It follows
  from Theorem~\ref{theorem:globalFormulasToAutomata} that we can
  construct in polynomial space a global MSCAs $\cG_i$ such that
  $\lang(\varphi_i)=\lang(\cG_i)$ and the size of $\cG_i$ is linear in
  the size of $\varphi_i$ for every $i\in[n]$. By
  Theorem~\ref{theorem:linearizationsAutomaton}, we can construct, for
  every $i\in[n]$, a 2APA $\cP_i$ such that, for all MSCs $M$ and
  $B$-bounded linearizations $W$ of $M$, we have
  $M\in\lang(\varphi_i)$ if and only if $W_B\in\lang(\cP_i)$. By
  simple inspection of the construction of
  Sect.~\ref{subsec:from_mscas_to_word_automata}, one can see that
  $\cP_i$ can be obtained in polynomial space. The number of states of
  $\cP_i$ is also polynomial.
  Using standard automata constructions for alternating automata, we
  can combine the automata $\cP_1,\ldots,\cP_n$ according to the
  construction of $\varphi$ to obtain a 2APA $\cP_\varphi$ such that,
  for all MSCs $M$ and $B$-bounded linearizations $W$ of $M$, we have
  $W_B\in\lang(\cP_\varphi)$ if and only if $M\in\lang(\varphi)$. This
  can be accomplished in polynomial space and the number of states of
  $\cP_\varphi$ is also polynomial in $B$ and the size of
  $\varphi$. Clearly, $\varphi$ is satisfiable by an existentially
  $B$-bounded MSC if and only if $\lang(\cP)$ is non-empty. Hence, by
  Prop.~\ref{prop:emptiness}, the satisfiability problem of $\varphi$
  can be decided in polynomial space.
  The hardness result follows from the
  PSPACE-hardness of the satisfiability problem of LTL.
\end{proof}

\section{The Model Checking Problem}\label{sec:modelChecking}

A communicating finite-state machine (also known as message-passing
automaton) is well suited to model the behavior of a distributed
system. It consists of a finite number of finite automata
communicating using order-preserving channels. To be more precise, we
recapitulate the definition from
\cite{DBLP:journals/corr/abs-1007-4764}.
\begin{defi}
  A \emph{communicating finite-state machine} (or \emph{CFM} for
  short) is a structure $\Cc=(H,(\mathcal{T}_p)_{p\in\mathbb{P}},F)$ where
  \begin{iteMize}{$\bullet$}
  \item $H$ is a finite set of \emph{message contents},
  \item for every $p\in\mathbb{P}$, $\mathcal{T}_p=(S_p,\to_p,\iota_p)$ is a
    finite labelled transition system over the alphabet
    $\Sigma_p\times H$ (i.e., $\mathop{\to_p}\subseteq
    S_p\times\Sigma_p\times H\times S_p$) with initial state
    $\iota_p\in S_p$,
  \item $F\subseteq\prod_{p\in\mathbb{P}}S_p$ is a set of global final
    states.
  \end{iteMize}
  Let $\Cc$ be a CFM and $M$ be an MSC. A \emph{run} of $\Cc$ on $M$
  is a pair $(\zeta,\chi)$ of mappings
  $\zeta:V^M\to\bigcup_{p\in\mathbb{P}}S_p$ and $\chi:V^M \to H$ such that,
  for all $v\in V^M$,
  \begin{iteMize}{$\bullet$}
  \item $\chi(v)=\chi(v')$ if there exists $v'\in V^M$ with
    $\eta_M(v,v')=\msg$,
  \item
    $(\zeta(v'),\lambda(v),\chi(v),\zeta(v))\in\mathop{\to}_{P_M(v)}$
    if there exists $v'\in V^M$ with $\eta_M(v',v)=\proc$, and
    $(\iota_p,\lambda(v),\chi(v),\zeta(v))\in\mathop{\to}_{P_M(v)}$
    otherwise.
  \end{iteMize}
  Let $\mathsf{cofin}_\zeta(p)=\{s\in S_p\mid\forall v\in V_p^M\exists
  v'\in V_p^M:v\prec_p^M v' \land \zeta(v')=s\}$. The run $(\zeta,\chi)$ is
  \emph{accepting} if there is some $(s_p)_{p\in\mathbb{P}}\in F$ such that
  $s_p\in\mathsf{cofin}_\zeta(p)$ for all $p\in\mathbb{P}$. The
  \emph{language} of $\Cc$ is the set $\lang(\Cc)$ of all MSCs $M$ for
  which there exists an accepting run.
\end{defi}
We now demonstrate that the bounded model checking problem for CFMs
and CRPDL is PSPACE-complete.
\begin{thm}\label{theorem:modelChecking}
  The following problem is PSPACE-complete:
  \begin{quote}
    \noindent Input: $B\in\N$ (given in unary), CFM $\Cc$,  and a global CRPDL formula $\varphi$.
  
    \noindent Question: Is there an existentially $B$-bounded MSC
  $M\in\lang(\Cc)$ with $M\models\varphi$?
  \end{quote}
\end{thm}
\begin{proof}
  In \cite{DBLP:journals/corr/abs-1007-4764}, it was shown that one
  can construct in polynomial space a Büchi automaton $\cB_\Cc$ from
  $\Cc$ which recognizes exactly the set of all $B$-bounded
  linearizations of the MSCs from $\lang(\Cc)$. Its number of states
  is polynomial in the maximal number of local states a transition
  system of~$\mathcal{C}$ has and exponential in $B$. In the proof of
  Theorem~\ref{theorem:satisfiability}, we already constructed in
  polynomial space a B\"uchi automaton $\cB_\varphi$ of exponential
  size accepting the set of all $B$-bounded linearizations of the MSCs
  satisfying $\varphi$. Hence, the model checking problem can be
  decided in polynomial space. The PSPACE-hardness follows from the
  PSPACE-hardness of the satisfiability problem.
\end{proof}
\begin{rem}
  The model checking problem for CRPDL and high-level message sequence
  charts (HMSCs) asks, given an HMSC $\mathcal{H}$ and a global CRPDL
  formula~$\varphi$, is there an MSC $M\in\lang(\mathcal{H})$ with
  $M\models\varphi$.  Using techniques from
  \cite{DBLP:journals/corr/abs-1007-4764} and the ideas from the proof
  of Theorem~\ref{theorem:modelChecking}, it can be shown that this
  problem is also PSPACE-complete.
\end{rem}

\section{Open Questions}

It is an interesting open question whether the bounded model checking
problem of CFMs and CRPDL enriched with the intersection
operator~\cite{HKT00,DBLP:journals/corr/abs-1007-4764} is still in
PSPACE.  It also needs to be investigated whether PDL is a proper
fragment of CRPDL and if CRPDL and global MSCAs are expressively
equivalent.  Furthermore, we would like to know more about the
expressive power of CRPDL and global MSCAs in general, especially in
comparison with the existential fragment of monadic second-order logic
(EMSO).

\bibliographystyle{abbrv}
\bibliography{lit}

\end{document}